\documentclass{amsart}
\usepackage{hyperref}
\usepackage{graphicx}
\usepackage{curves}
\usepackage{enumerate}
\newcommand{\arxiv}[1]{\href{http://arxiv.org/abs/#1}{arXiv:#1}}
\newcommand*{\mailto}[1]{\href{mailto:#1}{\nolinkurl{#1}}}

\newtheorem{theorem}{Theorem}[section]
\newtheorem{lemma}[theorem]{Lemma}
\newtheorem{corollary}[theorem]{Corollary}
\newtheorem{remark}[theorem]{Remark}

\newcommand{\R}{\mathbb{R}}

\newcommand{\C}{\mathbb{C}}
\newcommand{\T}{\mathbb{T}}

\newcommand{\noprint}[1]{}

\newcommand{\nn}{\nonumber}
\newcommand{\beq}{\begin{equation}}
\newcommand{\eeq}{\end{equation}}
\newcommand{\bea}{\begin{eqnarray}}
\newcommand{\eea}{\end{eqnarray}}

\newcommand{\ol}{\overline}
\newcommand{\pa}{\partial}
\newcommand{\ti}{\tilde}

\newcommand{\abs}[1]{\lvert#1 \rvert}

\newcommand{\clos}{\mathop{\mathrm{clos}}}

\newcommand{\id}{\mathbb{I}}
\newcommand{\I}{\mathrm{i}}
\newcommand{\E}{\mathrm{e}}

\newcommand{\re}{\mathop{\mathrm{Re}}}
\newcommand{\im}{\mathop{\mathrm{Im}}}
\newcommand{\sech}{\mathop{\mathrm{sech}}}
\DeclareMathOperator{\res}{Res}
\DeclareMathOperator{\Ai}{Ai}

\makeatletter
\newcommand{\dlmf}[1]{%
\cite[%
  \def\nextitem{\def\nextitem{, }}%
  \@for \el:=#1\do{\nextitem\href{http://dlmf.nist.gov/\el}{(\el)}}%
]{dlmf}%
}
\makeatother

\newcommand{\la}{\lambda}
\newcommand{\ga}{\gamma}

\numberwithin{equation}{section}

\newcommand{\rI}{\begin{pmatrix}  1 & 1 \end{pmatrix}}

\begin{document}

\title[Rarefaction Waves of the KdV Equation]{Rarefaction Waves of the Korteweg--de Vries Equation
via Nonlinear Steepest Descent}

\author[K. Andreiev]{Kyrylo Andreiev}
\address{B. Verkin Institute for Low Temperature Physics \\ 47, Lenin ave\\ 61103 Kharkiv\\ Ukraine}
\email{\href{mailto:kijjeue@gmail.com}{kijjeue@gmail.com}}

\author[I. Egorova]{Iryna Egorova}
\address{B. Verkin Institute for Low Temperature Physics\\ 47, Lenin ave\\ 61103 Kharkiv\\ Ukraine}
\email{\href{mailto:iraegorova@gmail.com}{iraegorova@gmail.com}}

\author[T. L. Lange]{Till Luc Lange}
\address{Faculty of Mathematics\\ University of Vienna\\
Oskar-Morgenstern-Platz 1\\ 1090 Wien\\ Austria}

\author[G. Teschl]{Gerald Teschl}
\address{Faculty of Mathematics\\ University of Vienna\\
Oskar-Morgenstern-Platz 1\\ 1090 Wien\\ Austria\\ and International Erwin Schr\"odinger
Institute for Mathematical Physics\\ Boltzmanngasse 9\\ 1090 Wien\\ Austria}
\email{\href{mailto:Gerald.Teschl@univie.ac.at}{Gerald.Teschl@univie.ac.at}}
\urladdr{\href{http://www.mat.univie.ac.at/~gerald/}{http://www.mat.univie.ac.at/\string~gerald/}}

\thanks{J. Differential Equations J. Differential Equations {\bf 261}, 5371--5410 (2016)}
\thanks{Research supported by the Austrian Science Fund (FWF) under Grants No.\ V120 and W1245.}

\begin{abstract}
We apply the method of nonlinear steepest descent to compute the long-time
asymptotics of the Korteweg--de Vries equation with steplike initial data leading to a rarefaction wave.
In addition to the leading asymptotic we also compute the next term in the asymptotic expansion of the
rarefaction wave, which was not known before. 
\end{abstract}

\maketitle


\section{Introduction}

In this paper we investigate the Cauchy problem for the Korteweg--de Vries (KdV) equation
\beq\label{kdv}
q_t(x,t)=6q(x,t)q_x(x,t)-q_{xxx}(x,t), \quad (x,t)\in\R\times\R_+,
\eeq
with steplike initial data $q(x,0)=q_0(x)$ satisfying
\beq \label{ini}
\left\{ \begin{array}{ll} q_0(x)\to 0,& \ \ \mbox{as}\ \ x\to +\infty,\\
q_0(x)\to c^2,&\ \ \mbox{as}\ \ x\to -\infty.\end{array}\right.
\eeq
This case is known as rarefaction problem. The corresponding long-time asymptotics of $q(x,t)$ as $t\to\infty$ are well understood
on a physical level of rigor (\cite{Zakh, LN, N}) and can be split into three main regions:
\begin{itemize}
\item In the region $x<-6 c^2 t$ the solution is asymptotically close to the background $c^2$.
\item In the region $-6c^2 t<x<0$ the solution can asymptotically be described by $-\frac{x}{6 t}$.
\item In the region $0 < x$ the solution is asymptotically given by a sum of solitons.
\end{itemize}
This is illustrated in Figure~\ref{fig:num}. For the corresponding shock problem we refer to \cite{Bik1, EGKT, gp1, gp2, KKt, LN2, Ven}.
\begin{figure}
\includegraphics[width=8cm]{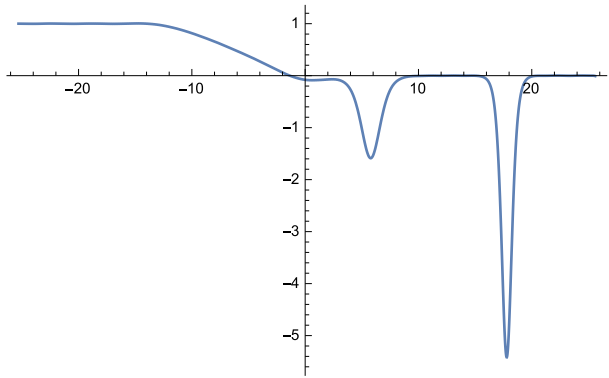}
\caption{Numerically computed solution $q(x,t)$ of the KdV equation at time $t=1.5$, with initial
condition $q_0(x)=\frac{1}{2}(1-\mathrm{erf}(x))-4\sech(x-1)$.} \label{fig:num}
\end{figure}

The aim of the present paper is to rigorously justify these results. Furthermore, we will also compute the second terms in the asymptotic expansion,
which were, to the best of our knowledge, not obtained before. Our approach is based on the nonlinear steepest descent method for oscillatory
Riemann--Hilbert (RH) problems. In turn, this approach rests on the inverse scattering transform for steplike initial data originally developed by
Buslaev and Fomin \cite{BF} with later contributions by Cohen and Kappeler \cite{CK}. For recent developments and further information we refer to \cite{EGLT}.
The application of the inverse scattering transform to the problem \eqref{kdv}--\eqref{ini}
(see \cite{EGT}, \cite{ET}) implies that the solution $q(x,t)$ of the Cauchy problem exists in the classical sense and is unique in the class
\beq\label{sol}
\int_0^\infty |x|(|q(x,t)| + |q(-x,t)-c^2|)dx<\infty, \qquad \forall t\in\R,
\eeq
provided the initial data satisfy the following conditions: $q_0\in \mathcal C^8(\R)$ and
\beq\label{decay1}
\int_0^{\infty} x^4\left( |q_0(x)| + |q_0(-x)-c^2| + |q^{(j)}(x)| \right)dx<\infty,\quad j=1,\dots,8.
\eeq
To simplify considerations we will additionally suppose that the initial condition decays exponentially fast to the asymptotics:
\beq\label{decay}
\int_0^{+\infty} \E^{\kappa  x}(|q_0(x)| + |q_0(-x)-c^2|)dx<\infty,
\eeq
for some small $\kappa>0$. We remark that by \cite{Ryb} the solution will be even real analytic under this assumption, but we will not need this fact.

This last assumption can be removed using analytic approximation of the reflection coefficient as demonstrated by Deift and Zhou \cite{dz} (see also \cite{GT,len}),
but we will not address this in the present paper. However, we emphasize that all known results concerning the asymptotic behavior of steplike solutions 
were obtained for the case of pure step initial data ($q_0(x)=0$ for $x>0$ and $q_0(x)=\pm c^2$ for $x\leq 0$) only. Moreover, those using the Riemann-Hilbert approach
did not address the parametrix problem, which is one of the main contributions of the present paper.

As is known, the solution of the initial value problem \eqref{kdv}, \eqref{decay1} can be computed by the inverse scattering transform from the right scattering data of the initial profile.
Here the right scattering data are given by the reflection coefficient $R(k)$, $k\in\R$, a finite number of eigenvalues $-\kappa_1^2$, \dots,
$-\kappa_N^2$, and positive norming constants $\gamma_1, \dots, \gamma_N$. 
The difference with the decaying case $c=0$ consists of the fact, that the modulus of the reflection coefficient is equal to 1 on the interval $[-c, c]$.
At the point $k=0$ the reflection coefficient takes the values $\pm 1$ (cf.\ \cite{CK}). The case $R(0)=-1$ known as the nonresonant case (which is generic),
whereas the case $R(0)=1$ is called the resonant case.
Note, that the right transmission coefficient $T(k)$ can be reconstructed uniquely from these data (cf.\ \cite{BF}). 

Our main results is the following

\begin{theorem} \label{maintheor} Let the initial data $q_0(x)\in\mathcal C^8(\R)$ of the Cauchy problem \eqref{kdv}--\eqref{ini} satisfy \eqref{decay}.
Let $q(x,t)$ be the solution of this problem. Then for arbitrary small $\epsilon_j>0$, $j=1,2,3$, and for $\xi=\frac{x}{12 t}$, the following asymptotics are valid as $t\to\infty$
uniformly with respect to $\xi$:

\noindent {\bf A}.
In the domain $(-6c^2 +\epsilon_1)t<x<-\epsilon_1 t$:
\beq\label{asq}
q(x,t)=-\frac{x + Q(\xi)}{6 t}(1 + O(t^{-1/3})), \quad \mbox{as}\  t\to +\infty,
\eeq
where
\beq\label{coef}
Q(\xi)=\frac{2}{\pi}\int_{-\sqrt{-2\xi}}^{\sqrt{-2\xi}}\left(\frac{d}{ds} \log R(s)-4\I\sum_{j=1}^N \frac{\kappa_j}{s^2 +\kappa_j^2}\right)\frac{ds}{\sqrt{s^2 +2\xi}} \mp \frac{1}{2\sqrt{-2\xi}},
\eeq
with $\pm$ corresponding to the resonant/nonresonant case, respectively.

\noindent {\bf B}.
In the domain $x<(-6c^2 -\epsilon_2)t$ in the nonresonant case:
\beq
q(x,t)=c^2 + \sqrt{\frac{4\nu  \tau }{3t}}\sin(16t\tau^3-\nu\log(192 t \tau^3)+\delta)(1+o(1)),
\eeq
where $\tau=\tau(\xi)=\sqrt{\frac{c^2}{2}-\xi}$, $\nu=\nu(\xi)=  -\frac{1}{2\pi} \log\left(1-|R(\tau)|^2\right)$ and
\begin{align*}
\delta(\xi)= & -\frac{3\pi}{4}+ \arg(R(\tau)-2T(\tau)+\Gamma(\I\nu))\\
&  -\frac{1}{\pi}\int_{\R\setminus[-\tau, \tau]}\log\frac{1-|R(s)|^2}{1-|R(\tau)|^2}\,\frac{s\, ds}{s^2 - c^2 -(\frac{c^2}{2} + \xi)^{1/2}(c^2 - s^2)^{1/2}}.
\end{align*}
 Here $\Gamma$ is the Gamma function.

\noindent {\bf C}. In the domain $x>\epsilon_3 t$:
\[
q(x,t)= -\sum_{j=1}^N \frac{2\kappa_j^2}{\cosh^2\left(\kappa_j x- 4\kappa_j^3 t-\frac{1}{2}\log\frac{\gamma_j}{2\kappa_j}-\sum_{i=j+1}^N\log\frac{\kappa_i - \kappa_j}{\kappa_i + \kappa_j}\right)} +O(\E^{-\epsilon_3 t/2}).
\]
\end{theorem}
We should remark that our results do not cover the two transitional regions: $0 \approx x$ near the leading wave front,
and $x\approx-6c^2 t$ near the back wave front. Since the error bounds
obtained from the RH method break down near these edges, a rigorous justification is beyond the scope of the present paper.

Our paper is organized as follows: Section~\ref{sec:rhp} provides some necessary information about the inverse scattering transform with steplike
backgrounds and formulates the initial vector RH problems. In Section~\ref{sec:sr} we study the soliton region.
In Section~\ref{secg1} the initial RH problem is reduced to a "model" problem in the domain $-6c^2 t<x<0$. It is solved in Section \ref{model},
and the question of a suitable parametrix is discussed in Section \ref{parametrix}.
Justification of the asymptotical formula \eqref{asq}--\eqref{coef} is given in Section~\ref{asymptotics}.
In Section~\ref{sec:left} we establish the asymptotics in the dispersive region $x<-6c^2 t $.

\section{Statement of the RH problem and the first conjugation step}
\label{sec:rhp}

Let $q(x,t)$ be the solution of the Cauchy problem \eqref{kdv}, \eqref{decay1} and consider the underlying spectral problem
\beq\label{Sp}
(H(t) f)(x):=- \frac{d^2 }{dx^2}f(x) + q(x,t)f(x)=\lambda f(x),\qquad x\in\R.
\eeq
In order to set up the respective Riemann--Hilbert (RH) problems we recall some facts from
scattering theory with steplike backgrounds. We refer to \cite{EGLT} for proofs and further details.

Throughout this paper we will use the following notations: Set $\mathcal D:=\C\setminus\Sigma$, where
\[
\Sigma=\Sigma_U\cup\Sigma_L,\quad \Sigma_U=\{\la^U=\la+\I 0, \la\in[0,\infty)\},\quad \Sigma_L=\{\la^L=\la-\I 0, \la\in[0,\infty)\}
\]
(throughout this paper the indices $U$ and $L$ will stand for "upper" and "lower").
That is, we treat the boundary of the domain $\mathcal D$ as consisting of the two sides of the cut along the interval $ [0,\infty)$,
with different points $\la^U$ and $\la^L$ on different sides. In equation \eqref{Sp} the spectral parameter $\la$ belongs to the set $\clos(\mathcal D)$, where $\clos(\mathcal D)=\mathcal D\cup\Sigma_U\cup\Sigma_L$.
Along   with  $\la$ we will use two more spectral parameters
\[
k=\sqrt \la, \quad k_1=\sqrt{\la - c^2},\quad \mbox{where}\ \ k>0\ \ \mbox{and}\ \  k_1>0,\ \  \mbox{ for }\ \la^U>c^2.
\]
The functions $k_1(\la)$ and $k(\la)$   conformally map the domain $\mathcal D$  onto $\mathfrak D_1:=\C^+\setminus (0, \I c]$ and $\mathfrak D:= \C^+$, respectively.
Since there is a bijection between the closed domains $\clos\mathcal D$, $\clos \mathfrak D=\C^+\cup \R$ and $\clos\mathfrak D_1=\mathfrak D_1\cup\R\cup [0,\I c]_r\cup [0,\I c]_l$, we will use the ambiguous notation $f(k)$ or $f(k_1)$ or $f(\la)$ for the same value of an arbitrary function $f(\la)$ in these respective coordinates.
Here the indices $l$ and $r$ are associated with  the right and  the left sides of the cut.  In particular, if $k>0$ corresponds to $\la^U$ then $-k$ corresponds to $\la^L$, and for functions defined on the set $\Sigma$ we will sometimes use the notation $f(k)$ and $f(-k)$ to indicate the values at symmetric points $\la^U$ and $\la^L$.

Since the potential $q(x,t)$ satisfies \eqref{sol}, the following facts are valid for the operator $H(t)$ (\cite{EGLT}):

\begin{theorem}\label{thm:scat}
{\ }
\begin{itemize}
\item The spectrum of $H(t)$ consists of an absolutely continuous part $\R_+$ plus a finite number of negative eigenvalues $\la_1<\cdots<\la_N<0$.
The (absolutely) continuous spectrum consists of a part $[0, c^2]$ of multiplicity one and a part $[c^2, \infty)$ of multiplicity two. In terms of the variables $k$ and $ k_1$, the continuous spectrum corresponds to $k\in \R$, and the spectrum of multiplicity two to $k_1\in \R$.

\item Equation \eqref{Sp} has two Jost solutions $\phi(\la,x,t)$ and $\phi_1(\la,x,t)$, satisfying the conditions
\[
\lim_{x \to  +\infty} \E^{-\I k x} \phi(\la,x,t) = \lim_{x \to  -\infty} \E^{\I k_1 x} \phi_1(\la,x,t) =1,\quad \mbox{for}\ \ \la\in\clos\mathcal D.
\]
The Jost solutions fulfill the scattering relations
\begin{align}
  \label{rsr}
  T(\la,t) \phi_1(\la,x,t)
  & =\overline{\phi(\la,x,t)} + R(\la,t)\phi(\la,x,t), \qquad k\in\R, \\  \label{lsr}
  T_1(\la,t) \phi(\la,x,t)
   &=\overline{\phi_1(\la,x,t)} + R_1(\la,t)\phi_1(\la,x,t), \quad k_1\in\R,
\end{align}
where $T(\la,t)$, $R(\la,t)$ (resp., $T_1(\la,t)$, $R_1(\la,t)$) are the right (resp., the left) transmission and reflection coefficients.

\item The Wronskian 
\beq\label{wrons}W(\la,t)=\phi_1(\la,x,t)\phi^\prime(\la,x,t) -\phi_1^\prime(\la,x,t)\phi(\la,x,t)\eeq of the Jost solutions has simple zeros at the points $\la_j$.
The only other possible zero is $\lambda=0$. The case $W(0,t)=0$ is known as the resonant case. In this case $R(0,t)=1$.
In the nonresonant case, which is generic, $R(0,t)=-1$.

\item The solutions $\phi(\la_j,x,t)$ and $\phi_1(\la_j,x,t)$ are the corresponding (linearly dependent) eigenfunctions of $H(t)$. The associated norming
constants are
\beq\label{rtime7}
\gamma_j(t)=\left(\int_\R \phi^2(\la_j,x,t) dx \right)^{-1}, \quad \gamma_{j,1}(t)=\left(\int_\R \phi^2_1(\la_j,x,t) dx \right)^{-1}.
\eeq
\item The function $T(\la,t)$  has a meromorphic extensions
to the domain $\la\in\C\setminus[0,\infty)$ with simple poles at the points $\la_1$,\dots, $\la_N$.
The only possible zero is at $\la=0$ in the nonresonant case. In the resonant case $T(\la,t)\neq 0$ for all $\la\in\clos(\mathcal D)$.
\item There is a symmetry $T(\la^U,t)=\overline{T(\la^L,t)}$, $T_1(\la^U,t)=\overline{T_1(\la^L,t)}$, $R(\la^U,t)=\overline{R(\la^L,t)}$ for $k\in\R$, i.e.\ for $\la\in\Sigma$.
The same is valid for $\phi(\la,x,t)$ and $\phi_1(\la,x,t)$. Moreover, $\phi_1(\la,x,t)\in \R$ for $k\in[-c,c]$ and
$R_1(\la^U,t)=\overline{R_1(\la^L,t)}$ for $k_1\in\R$.
\item The following identities are valid on the continuous spectrum:
\beq\label{RTR}
-\frac{T_1(\la,t)}{\ol{T_1(\la,t)}}=\frac{T(\la,t)}{\ol{T(\la,t)}}=R(\la,t)\E^{2 \I \arg k},\ \mbox{for}\ \  k\in[-c, c],
\eeq
and for $k_1\in\R$:
\begin{align}\label{RTR1}
&1-|R(\la,t)|^2 =1-|R_1(\la,t)|^2=T_1(\la,t)\ol{T(\la,t)},\\ \nn &  R_1(\la,t)\ol{T(\la,t)} +\ol{R(\la,t)}T(\la,t)=0.
\end{align}
Here $\arg k=\pi$ as $k<0$.
\item The spectrum is time independent and the time evolution of the scattering data is given by (\cite{ET,Kh1,Kh2})
\begin{align}\nn
R(\la,t)&=R(k)\E^{8\I k^3 t}, \quad\quad\qquad k\in\R, \\ \label{rtime}
\chi(\la,t)&=\chi(k)\E^{-8\I t k_1^3 - 12 \I t k_1 c^2},\qquad k\in [-c,c],\\ \nn
R_1(\la,t)&=R_1(k_1)\E^{-8\I t k_1^2 - 12\I t k_1c^2},\qquad k_1\in\R,\\ \nn
\gamma_{1,j}(t)&=\gamma_{1,j}\E^{-8\kappa_{1,j}^3  t +12 c^2\kappa_{1,j}t},\qquad\gamma_j(t)=\gamma_j\E^{8 \kappa_j^3 t},
\end{align}
where
\beq\label{rtime3}
\chi(\la,t):=-\ol{T_1(\la, t)}T(\la,t);\ \  \chi(k)=\chi(\la(k),0);\  R(k)=R(\la(k),0);
\eeq
\beq\label{rtime5}
\gamma_{1,j}=\gamma_{1,j}(0);\ \gamma_j=\gamma_j(0);\ 0<\kappa_j=\sqrt{-\la_j};\  
\kappa_{1,j}=\sqrt{\kappa_j^2 + c^2}.
\eeq
\item Under the assumption \eqref{decay} with $0<\kappa<\kappa_N$ the solution $\ol{\phi(\la,x,0)}$ has an analytical continuation to a subdomain $\mathcal D_\kappa\subset\mathcal D$, where $\mathcal D_\kappa=\{\la(k):\, 0< \im k<\kappa\}$.
Accordingly, the function $R(k)$ has a holomorphic continuation to  the strip $0<\im k <\kappa$, continuous up to the boundary $\im k=0$. The transmission coefficient as a function of $k$ always has an analytical continuation in $\C^+$, and is holomorphic in the strip and continuous up to the boundary $\R$.
Identity \eqref{rsr} remains valid in the strip. 
\item The solution $q(x,t)$ of the initial value problem \eqref{kdv}--\eqref{decay1}   can be uniquely recovered from either the right initial scattering data
\[
\{ R(k), \ \ k\in\R;\ \ \la_j=-\kappa_j^2,\ \gamma_j>0,\  j=1,\dots,N\},
\]
or from the left initial scattering data
\[
\{ R_1(k_1),\ \ k_1\in\R;\ \ \chi(k), \ k\in[-c,c];\ \  \la_j,\  \gamma_{1,j}>0,\ \  j=1,\dots,N\}.
\]
\end{itemize}
\end{theorem}

These properties allow us to formulate  two vector RH problems. One of them is connected with the right scattering data, another one with the left one.
To this end we introduce a vector function
\beq \label{defm}
m(\la,x,t) =
\begin{pmatrix} T(\la,t) \phi_1(\la,x,t) \E^{\I k x},  & \phi(\la,x,t)  \E^{-\I k x} \end{pmatrix}
\eeq
on $\clos \mathcal D$.
By Theorem~\ref{thm:scat} this function is meromorphic in $\mathcal D$ with simple poles at the points $\la_j$, and continuous up to the boundary $\Sigma$.
We regard it as a function of $k\in \ol{\C_+}$ (with $\ol{\C_+}$ the closed upper half-plane), keeping $x$ and $t$ as parameters. Accordingly we will write $m(k):=m(\la(k), x, t)$.
This vector function has the following asymptotical behavior (cf.\ \cite{EGKT} and \cite{EGLT}, Lemma 4.3) as $k\to\infty$ in any direction of $\ol{\C_+}$:
\beq\label{asm}
m(k) = \begin{pmatrix} m_1(k)& m_2(k)\end{pmatrix}=\rI -\frac{1}{2\I k}\int_{x}^{+\infty}q(y,t) dy \begin{pmatrix} -1 & 1 \end{pmatrix} + O\left(\frac{1}{k^2}\right),
\eeq
and
\beq\label{eq:asygf}
m_1(k) m_2(k) = T(k)\phi(k,x,t)\phi_1(k,x,t)=1+\frac{q(x,t)}{2k^2} + O(\frac{1}{k^4}).
\eeq

Extend the definition of $m(k)$ to $\C^-$ using the {\it symmetry condition}
\beq\label{symto}
m(k)=m(-k)\sigma_1,
\eeq
where
\[
\sigma_1=\begin{pmatrix} 0& 1\\ 1 & 0\end{pmatrix}, \quad
\sigma_2=\begin{pmatrix} 0 & -\I\\ \I& 0\end{pmatrix}, \quad
\sigma_3=\begin{pmatrix} 1& 0\\ 0 & -1\end{pmatrix}
\]
are the Pauli matrices.
After this extension the function $m$ has a jump along the real axis. We consider the real axis as a contour with the natural orientation from minus to plus infinity,
and denote by $m_+(k)$ (resp.\ $m_-(k)$) the limiting values of $m(k)$ from the upper (resp.\ lower) half-plane. 

\begin{theorem}\label{thm:vecrhp}
Let $
\{ R(k), \, k\in\R;\  \la_j=-\kappa_j^2,\ \gamma_j>0,\  j=1,\dots,N\}
$
be the right scattering data for the initial datum $q_0(x)$, satisfying condition \eqref{decay1},
and let $q(x,t)$ be the unique solution of the Cauchy problem \eqref{kdv}, \eqref{decay1}.
Then the vector-valued function $m(k)$ defined by \eqref{defm} and \eqref{symto}
is the unique solution of the following vector Riemann--Hilbert problem:

Find a vector-valued function $m(k)$, which is meromorphic  away from the
contour $\R$ with continuous limits from both sides of the contour and satisfies:
\begin{enumerate}[I.]
\item  The jump condition $m_{+}(k)=m_{-}(k) v(k)$, where
\beq \label{jump15}
v(k):=v(\la(k), x,t)= \begin{pmatrix}
1-\abs{R(k)}^2 & - \ol{R(k)} \E^{- 2 t \Phi(k)} \\[1mm]
R(k) \E^{2 t \Phi(k)} & 1
\end{pmatrix}, \quad  k\in \R;  \eeq
\item
the pole conditions
\beq\label{pole15}
\aligned
\res_{\I\kappa_{j}} m(k) &= \lim_{k\to \I\kappa_{j}} m(k)
\begin{pmatrix} 0 & 0\\ \I \ga_{j} \E^{2t\Phi(\I\kappa_{j})}  & 0 \end{pmatrix},\\
\res_{-\I \kappa_j} m(k) &= \lim_{k\to -\I\kappa_j} m(k)
\begin{pmatrix} 0 &  -\I \ga_{j} \E^{-2t\Phi(\I\kappa_j)} \\ 0 & 0 \end{pmatrix};
\endaligned
\eeq
\item
the symmetry condition \eqref{symto};
\item
the normalization condition
\beq\label{normcond}
 m(k) = \rI + O(k^{-1}),\quad k\to\infty.
\eeq
\end{enumerate}
Here the phase $\Phi(k)=\Phi(k,x,t)$ in \eqref{jump15}  is given by
\[
\Phi(k)= 4\I k^3 +\I k \frac{x}{t}.
\]
\end{theorem}

\begin{remark}
Note that by property \eqref{RTR} we have $|R(k)|=1$ for $k\in[-c,c]$ implying
\[
v(k)= \begin{pmatrix}
0 & - \ol{R(k)} \E^{- 2 t \Phi(k)} \\[1mm]
R(k) \E^{2 t \Phi(k)} & 1
\end{pmatrix}, \quad  k\in [-c,c].
\]
\end{remark}
 
\begin{proof}[Proof of Theorem \ref{thm:vecrhp}.]
Since our further considerations mainly affect $k$, we drop $x$ and $t$ from our notation whenever possible. Let $m(k)$ be defined by \eqref{defm}.
In the upper half-plane it is a meromorphic function, its first component $m_1(k)$ has simple poles at points $\I\kappa_j$, and the second component $m_2(k)$ is holomorphic one.
Both components have continuous limits up to the boundary $\R$,  moreover, for $k\in \R$ we have $m_+(-k)=\ol{m_+(k)}$.
To compute the jump condition we observe that if $m_+= \big( T\phi_1 z,\:\: \phi z^{-1}\big)$,
where $z=\E^{\I k x}$,  $k\in \R$, then by the symmetry condition
$m_-= \big( \ol{\phi} z,\:\: \ol{T\phi_1} z^{-1}\big)$ at the same point $k\in \R$.
Write $\big(\begin{smallmatrix}\alpha(k) &\beta(k)\\ \gamma(k) & \delta(k) \end{smallmatrix}\big)$
for the unknown jump matrix. Then
\[
T\phi_1 z  = \ol\phi \, z\alpha + \ol {T\phi_1} z^{-1}\gamma,\qquad
\phi\, z^{-1}=\ol\phi\, z\beta + \ol{ T\phi_1}z^{-1}\delta.
\]
Multiply the first equality by $z^{-1}$, the second one by $z$, and then conjugate both of them.
We finally get
\beq\label{vnutr}
\ol\alpha\phi=\ol{T\phi_1} -T\ol{\gamma}\phi_1 z^{2},\qquad T\ol\delta\phi_1  = \ol\phi -\ol{\beta}\phi z^{-2}.
\eeq
Now divide the first of these equalities by $\ol T$ and compare it with \eqref{lsr} as $k_1\in\R$. From \eqref{RTR1} it follows that $\alpha=T_1 \ol T=1-|R|^2$, $\gamma z^{-2}=R$ if $k_1\in\R$.
For $k\in [-c,c]$ we use the first equality of \eqref{vnutr} taking into account that $\ol\phi_1=\phi_1$. Then by \eqref{RTR} $\ol\alpha\phi =\phi_1\ol T(1-\ol \gamma z^2 R)$  and
therefore $\alpha=0$, $\gamma z^{-2}= R$ if $k\in[-c,c]$. Taking into account \eqref{rtime}  and $z=\E^{\I k x}$ we finally justify the $11$ and $21$ entries of the jump matrix \eqref{jump15}. Comparing the second equality of \eqref{vnutr} with \eqref{rsr} gives $\delta=1$ and $-\ol\beta z^{-2}=R$. This justifies the $12$ and $22$ entries.

The pole condition \eqref{pole15} is proved in \cite{GT} or in Appendix~A of \cite{EGKT}.
The symmetry condition holds by definition, and the normalization condition follows from \eqref{asm}.

It remains to prove that the solution of this RH problem is unique. Let $m(k)$ and $\tilde m(k)$ be two solutions. Then $\hat m(k)=m(k) - \tilde m(k)$ satisfies I--III (note, that condition II does not guarantee that $\hat m$ is a holomorphic solution!) and condition IV is replaced by $\hat m(k)=O(k^{-1})$. Therefore the function
\[
F(k):=\hat m_1(k)\ol{\hat m_1(\ol {k})} +\hat m_2(k)\ol{\hat m_2(\ol {k})}
\]
is a meromorphic in $\C^+$ with simple poles at the points $\I \kappa_j$ and with the asymptotical behavior $F(k)=O(k^{-2})$ as $k\to\infty$.
Since $-\ol {k}=k$ for $k\in\I\R$ condition II implies
\beq\label{vychet}
\res_{\I \kappa_j} F(k)=2\I\gamma_j |\hat m_2(\I\kappa_j)|^2 \E^{2 t \Phi(\I\kappa_j)}\in\I \R_+.
\eeq
Moreover, $F(k)$ has continuous limiting values $F_+(k)$ on $\R$, which can be represented, due to condition III, as
$ F_+(k)=\hat m_{1,+}(k)\ol{\hat m_{1,-}(k)} + \hat m_{2,+}(k)\ol{\hat m_{2,-}(k)}.$
From condition I we then get
\begin{align*}
F_+(k) &= \big((1-|R|^2) \hat m_{1,-} + \mathcal R \hat m_{2,-} \big)\ol{\hat m_{1,-}} + \big(\hat m_{2,-} -\ol{\mathcal R} \hat m_{1,-} \big)\ol{\hat m_{2,-}}\\
&= (1 - |R|^2) |\hat m_{1,-}|^2 +  |\hat m_{2,-}|^2 + 2\I \im\big( \mathcal R \ol{\hat m_{1,-}} \hat m_{2,-} \big).
\end{align*}
Now let $\rho>\kappa_1$  and consider the half-circle
\[
\mathcal C_\rho=\{ k: k\in [-\rho, \rho], \mbox{ or } k=\rho\E^{\I \theta}, \quad 0<\theta<\pi\}
\]
as a contour, oriented counterclockwise. By the Cauchy theorem and \eqref{vychet}
\[
\oint_{\mathcal C_{\rho}} F(k)dk=2\pi\I\sum_{j=1}^N \res_{\I\kappa}F(k)=-4\pi \sum_{j=1}^N
\gamma_j |\hat m_2(\I\kappa_j)|^2 \E^{2 t \Phi(\I\kappa_j)}.
\]
Using $F(k)=O(k^{-2})$ as $k\to\infty$ we see $\lim_{\rho\to\infty }\int_0^\pi F(\rho\E^{\I\theta})\rho\E^{\I\theta}d\theta=0$ implying
\[
\int_\R F_+(k)dk + 4\pi \sum_{j=1}^N
\gamma_j |\hat m_2(\I\kappa_j)|^2 \E^{2 t \Phi(\I\kappa_j)}=0.
\]
Taking the real part we further obtain
\[
\int_\R \big( (1 - |R(k)|^2) |\hat m_{1,-}(k)|^2 +  |\hat m_{2,-}(k)|^2 \big)dk + 4\pi \sum_{j=1}^N
\gamma_j |\hat m_2(\I\kappa_j)|^2 \E^{2 t \Phi(\I\kappa_j)}=0,
\]
which shows
\[
\hat m_{2,-}(k)=0 \ \mbox{as}\  k\in(\R\cup_j \{\I\kappa_j\}), \ \mbox{and}\ \  \hat m_{1,-}(k)=0
\ \mbox{as}\  k_1\in\R.
\]
Thus, function $F(k)$ is entire and taking into account its behavior at infinity we conclude that it is zero.
This proves uniqueness of the RH problem under consideration.
\end{proof}

For our further analysis we rewrite the pole condition as a jump
condition, and hence we turn our meromorphic Riemann--Hilbert problem into a holomorphic one literally following \cite{GT}.
Choose $\delta>0$ so small that the discs $\left\vert k- \I \kappa_j \right\vert<\delta$ lie inside the upper half-plane and
do not intersect any of the other contours, moreover $\kappa_N-\delta>\kappa$,
where $\kappa$ is the same as in estimate \eqref{decay}.
Redefine $m(k)$ in a neighborhood of $\I \kappa_j$ (respectively $- \I \kappa_j$) according to
\beq\label{eq:redefm}
m(k) = \begin{cases} m(k) \begin{pmatrix} 1 & 0 \\
-\frac{\I \gamma_j \E^{2t\Phi(\I \kappa_j)} }{k- \I \kappa_j} & 1 \end{pmatrix},  &
|k- \I \kappa_j|< \delta,\\
m(k) \begin{pmatrix} 1 & \frac{\I \gamma_j \E^{2t\Phi(\I \kappa_j)} }{k+ \I \kappa_j} \\
0 & 1 \end{pmatrix},  &
|k+ \I \kappa_j|< \delta,\\
m(k), & \text{else}.\end{cases}
\eeq
Denote the boundaries of these small discs as $\T^{j,U}$ and $\T^{j,L}$ (as usual, indices $U$ and $L$ are associated  with "upper" and "lower").
Set also
\beq\label{remm}
h^{U}(k,j):= -\frac{\I \gamma_j \E^{2t\Phi(\I\kappa_j)}}{k-\I \kappa_j},\quad
h^{L}(k,j):= -\frac{\I \gamma_j \E^{2t\Phi(\I\kappa_j)}}{k+\I \kappa_j}.
\eeq
Then a straightforward calculation using $\res_{\I \kappa} m(k) = \lim_{k\to\I\kappa} (k-\I \kappa)m(k)$ shows the following well-known result (see \cite{GT}):

\begin{lemma}\label{lem:holrhp}
Suppose $m(k)$ is redefined as in \eqref{eq:redefm}. Then $m(k)$ is holomorphic in $\C\setminus\left(\R\cup\bigcup_{j=1}^N (\T^{j,U}\cup \T^{j,L})\right)$. Furthermore it satisfies conditions I, III, IV and II is replaced by the jump condition
\beq \label{eq:jumpcond2}
m_+(k) = m_-(k) \begin{cases} \begin{pmatrix} 1 & 0 \\
h^{U}(k,j) & 1 \end{pmatrix}, & k\in \T^{j,U}\\
 \begin{pmatrix} 1 & h^{L}(k,j)\\
0 & 1 \end{pmatrix},& k\in\T^{j,L},
\end{cases}
\eeq
where the small circles $\T^{j,U}$ around $\I \kappa_j$ are oriented counterclockwise, and the circles $\T^{j,L}$ around $-\I \kappa_j$ are oriented clockwise.
\end{lemma}

This "holomorphic" RH problem is equivalent  to the initial one, given by conditions I--IV. Thus, it also has a unique solution. We use it everywhere except of small regions of
$(x,t)$ half-plane in vicinities of the rays  $x=4\kappa_j^2 t$, which correspond to the solitons. In what follows we will denote this RH problem as RH-$k$ problem, associated with the right scattering data. This problem is convenient for investigations in the region $x> - 6 c^2 t$.  In the remaining  region it
turns out more convenient to use an RH-$k_1$ problem, associated with the left scattering data. In this region we study the nonresonant case only. 

Let $\mathfrak D_1=\C^+\setminus (0,\I c]$ be the domain for $k_1$, which is in one-to-one correspondence with the domain $\mathcal D$ for $\lambda$ as well as with the upper half-plane for $k$. As already pointed out before we will simply consider the scattering data and Jost solutions as functions of $k_1$.

 In the $\C$ plane of the $k_1$ variable we consider the cross contour consisting of the real axis $\R$, with the  orientation from minus to plus infinity, and of the vertical segment $[\I c, -\I c]$, oriented top-down. The images of the discrete spectrum of $H(t)$ are now located at the points $\pm\I\kappa_{1,j}$, $\kappa_{1,j}>c$ (see Theorem~\ref{thm:scat}, formulas \eqref{rtime7}, \eqref{rtime3}, \eqref{rtime5}). By definition, $\chi(k)$, considered as the function of $k_1$, is defined on the contour $[\I c, 0]$ as
$
\chi(k_1)=-\ol{T_1(k_1,0)}T(k_1,0),\ \mbox{as}\  k_1\in [0,\I c]_r, \ \mbox{i.e.}\ k\in[0,c].
$ 
We define it on $[0,-\I c]$ as
$
\chi(-k_1):=-\chi(k_1).
$ In the nonresonant case this is a continuous function for $k_1\in [-\I c, \I c]$ with $\chi(-\I c)=\chi(\I c)=\chi(0)=0$.

In $\mathfrak D_1$ we introduce the vector-valued function
\beq\label{defmn}
m^{(1)}(k_1,x,t)=
\left( T_1(k_1,t) \phi(k_1,x,t) \E^{-\I k_1 x} ,\  \phi_1(k_1,x,t) \E^{\I k_1 x} \right)
\eeq
and extend it to the lower half-plane by the symmetry condition
\beq\label{symto2}
m^{(1)}(-k_1)=m^{(1)}(k_1)\sigma_1.
\eeq
In the nonresonant case this vector function has continuous limits on the boundary of the domain $\mathcal D_1$ and 
 has the following asymptotical behavior as $k_1\to \infty$:
\beq\label{asm1}
m^{(1)}(k_1,x,t)= \rI +\frac{1}{2\I k_1}\left(\int^x_{-\infty}(q(y,t)- c^2)dy\right) \begin{pmatrix} 1 & -1 \end{pmatrix} + O\left(\frac{1}{k_1^2}\right).
\eeq

\begin{theorem}\label{thm:vecrhpn}
Let $\{ R_1(k_1),\; k_1\in \R; \  \chi(k_1),\ k_1\in [0,\I c];\ (\kappa_{1,j}, \ga_{1,j}), \: 1\le j \le N \}$ be
the left scattering data of the operator $H(0)$ which correspond to the nonresonant case. Let $\T_j^U$ (resp., $\T_j^L$) be circles with centres in $\I \kappa_{1,j}$ (resp., $-\I \kappa_{1,j}$) and with radii $0<\varepsilon<\frac{1}{4}\min_{j=1}^N|\kappa_{1,j} - \kappa_{1,j-1}|,$ $\kappa_{1,0}=0$. Then $m^{(1)}(k_1)=m^{(1)}(k_1,x,t)$, defined in \eqref{defmn}, \eqref{symto2},
is the unique solution of the following vector Riemann--Hilbert problem:
Find a vector function $m^{(1)}(k_1)$ which is holomorphic away from the contour $\bigcup_{j=1}^N (\T_j^U\cup\T_j^L)\cup\R\cup [-\I c, \I c]$,
has continuous limiting values from both sides of the contour and satisfies:
\begin{enumerate}
\item The jump condition $m_+^{(1)}(k_1)=m_-^{(1)}(k_1) v^{(1)}(k_1)$
\beq\label{vleft}
v^{(1)}(k_1)=\left\{\begin{array}{cc}\begin{pmatrix}
1-|R_1(k_1)|^2 & - \ol{R_1(k_1)} \E^{- 2t\Phi_1(k_1)} \\
R_1(k_1) \E^{2 t\Phi_1(k_1)} & 1
\end{pmatrix},& k_1\in\R,\\
 \ &\ \\
\begin{pmatrix}
1 & 0 \\
\chi(k_1) \E^{2 t\Phi_1(k_1)} & 1
\end{pmatrix},& k_1\in [ \I c, 0],\\
\ &\ \\
\begin{pmatrix}
1 &  \chi(k_1) \E^{-2 t\Phi_1(k_1)} \\
0 & 1
\end{pmatrix},& k_1\in [ 0, -\I c],\\
\ &\ \\
 \begin{pmatrix} 1 & 0 \\
-\frac{\I \gamma_{1,j} \E^{t\Phi_1(\I\kappa_{1,j})}}{k_1-\I \kappa_{1,j}} & 1 \end{pmatrix},& k_1\in \T_j^U,\\
\ &\ \\
 \begin{pmatrix} 1 & -\frac{\I \gamma_{1,j} \E^{-t\Phi_1(-\I \kappa_{1,j})}}{k_1+ \I \kappa_{1,j}} \\
0 & 1 \end{pmatrix},& k_1\in\T_j^L;
\end{array}\right.
\eeq
\item
the symmetry condition
\eqref{symto2};
\item
the normalization condition
$
\lim_{\kappa\to\infty} m^{(1)}(\I\kappa) = (1\quad 1).
$
\end{enumerate}
Here the phase $\Phi_1(k)=\Phi_1(k_1,x,t)$ is given by
\[
\Phi_1(k_1)= -4 \I k_1^3- 6\I c^2 k_1 -12\I\xi k_1,\quad \xi=\frac{x}{12 t},
\]
and the circles are oriented in the same way as in Lemma \ref{lem:holrhp}.
\end{theorem}

The proof of this theorem is analogous to the proof of Theorem \ref{thm:vecrhp}. 
\begin{remark}
In our above formulations of the RH problems we could replace the continuous limits by non-tangential $L^2$ limits (cf.\ \cite[Sect.~7.1]{deiftbook}).
Locally this is clear and globally this follows from the normalization conditions (which is supposed to hold around the contour as well).
All our RH problems will satisfy the stronger condition from above (except for possibly a finite number of points in the model problems later on) and
hence we have chosen to use this simpler formulation.
\end{remark}
Our first aim is to reduce these RH problems to model problems which can be solved explicitly.
To this end we record the following well-known result (see e.g.\ \cite{GT}) for easy reference.

\begin{lemma}[Conjugation] \label{lem:conjug}
Let $v(k)$ be a continuous matrix on the contour $\hat \Sigma$, where $\hat\Sigma$ is one of the contours which appeared in Theorem~\ref{thm:vecrhp} or \ref{thm:vecrhpn}.
Let $m(k)$, $k\in\C$, be a holomorphic solution of the RH problem $m_+(k)=m_-(k) v(k)$, $k \in\hat\Sigma$, which has continuous limiting values from both sides of the contour
and which satisfies the symmetry and normalization conditions. 
Let $\tilde \Sigma\subset\hat\Sigma$ be a contour with the same orientation. Suppose that  $\tilde\Sigma$
contains with each point $k$ also the point $-k$. Let $D$ be a matrix of the form
\beq\label{matrixD} D(k) = \begin{pmatrix} d(k)^{-1} & 0 \\ 0 & d(k) \end{pmatrix},
\eeq
where $d: \hat\C\setminus\tilde \Sigma\to\C$ is a sectionally analytic function with $d(k)\neq 0$ except for a
finite number of points on $\tilde\Sigma$. Set
\beq\label{mD}\ti{m}(k) = m(k) D(k),\eeq
then the jump matrix of the problem $\ti{m}_+=\ti{m}_- \ti{v}$ is
\[
\ti{v} =
\begin{cases}
  \begin{pmatrix} v_{11} & v_{12} d^{2} \\ v_{21} d^{-2}  & v_{22} \end{pmatrix}, &
  \quad k \in \hat \Sigma \backslash \tilde \Sigma, \\[6mm]
  \begin{pmatrix} v_{11} d_+^{-1} d_-   & v_{12} d_+ d_- \\
  v_{21} d_+^{-1} d_-^{-1}  & v_{22} d_-^{-1} d_+   \end{pmatrix}, &
  \quad k \in\tilde \Sigma.
\end{cases}
\]
If $d$ satisfies $d(-k) = d(k)^{-1}$ for $k \in\C\setminus\tilde \Sigma$,
then the transformation \eqref{mD} respects the symmetry condition \eqref{symto}.  If $d(k)\to 1$ as $k\to\infty$ then \eqref{mD} respects the normalization condition \eqref{normcond}.
\end{lemma}

Note that in general, for an oriented contour $\hat \Sigma$, the value $f_+(k_0)$ (resp.\ $f_-(k_0)$) will denote the nontangential limit
 of the vector function $f(k)$ as $k\to k_0 \in \hat \Sigma$ from the positive (resp.\ negative) side of $\hat \Sigma$, where the
 positive side is the one which lies to
the left as one traverses the contour in the
direction of its orientation.

\section{Soliton region, $x>0$.}
\label{sec:sr}

Here we use the holomorphic RH problem with jump given by \eqref{jump15}, \eqref{eq:jumpcond2}, and \eqref{remm}.
We consider $x$ and $t$ as parameters, which change in a way that the value $\xi=\frac{x}{12t}$ evolves slowly when $x$ and $t$ are sufficiently large.
In the region under consideration we have $\xi>0$. To reduce our RH problem to a model problem which can be solved explicitly, we will use the well-known conjugation and
deformation techniques (\cite{GT}, \cite{EGKT}).

The signature table of the phase function $\Phi(k)=4\I k^3 + 12\I \xi k$ in this region is shown in Figure~\ref{fig:Phi-sol}.

\begin{figure}\centering
\begin{picture}(7,5.2)
\put(0.6,2.5){\line(1,0){5.8}}
\put(2,2.5){\vector(1,0){0.4}}
\put(5,2.5){\vector(1,0){0.4}}

\put(6,2.2){$\R$}


\put(0.6,3.5){$\re(\Phi)<0$}
\put(0.6,1.3){$\re(\Phi)>0$}
\put(2.9,4.5){$\re(\Phi)>0$}
\put(2.9,0.2){$\re(\Phi)<0$}

\put(2.7,3.9){\small $T^{j,U}$}
\put(2.7,2.9){\small$T^{i,U}$}
\put(2.7,1.9){\small$T^{i,L}$}
\put(2.7,0.9){\small$T^{j,L}$}

\put(3.6,2.4){\line(0,1){0.2}}
\put(3.6,2.0){\circle*{0.1}}
\put(3.6,2.0){\circle{0.4}}
\put(3.6,3.0){\circle*{0.1}}
\put(3.6,3.0){\circle{0.4}}
\put(3.6,1.0){\circle*{0.1}}
\put(3.6,1.0){\circle{0.4}}
\put(3.6,4.0){\circle*{0.1}}
\put(3.6,4.0){\circle{0.4}}

\put(3.6,1.5){\circle*{0.09}}
\put(3.6,3.5){\circle*{0.09}}

\put(3.8,3.3){\small$\I \kappa_0$}
\put(3.8,1.5){\small$-\I \kappa_0$}
\curvedashes{0.05,0.05}

\curve(1.868,5., 2.1,4.625, 2.6,4., 3.1,3.625, 3.6,3.5, 4.1,3.625, 4.6,4., 5.1,4.625, 5.332,5.)

\curve(1.868,0., 2.1,0.375, 2.6,1., 3.1,1.375, 3.6,1.5, 4.1,1.375, 4.6,1., 5.1,0.375, 5.332,0.)
\end{picture}
\caption{Signature table for $\re\Phi(k)$ in the soliton region.}\label{fig:Phi-sol}
\end{figure}
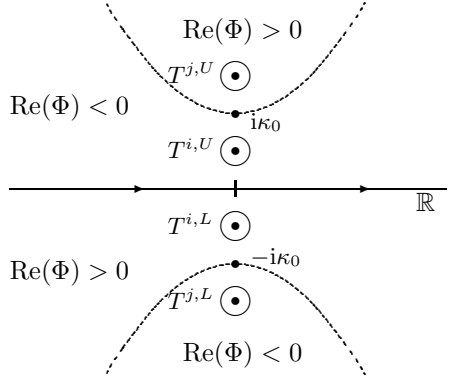

Namely, $\re \Phi(k)=0$ if $\im k=0$ or $(\im k)^2 - 3(\re k)^2=3\xi$, where the second curve consists of two hyperbolas which cross
the imaginary axis at the points $\pm\I\sqrt{3\xi}$. Set
\[
 \kappa_0 = \sqrt{\frac{x}{4 t}}>0.
\]
 Then we have $\re(\Phi(\I \kappa_j))>0$
for all $\kappa_j > \kappa_0$ and  $\re(\Phi(\I \kappa_j))<0$ for all $\kappa_j < \kappa_0$. Hence, in the first case
the off-diagonal entries of our jump matrices \eqref{eq:jumpcond2} are exponentially growing, and we need to turn them into exponentially
decaying ones. We set
\[
\Lambda(k,\xi):=\Lambda(k)=\prod_{\kappa_j > \kappa_0} \frac{k+\I\kappa_j}{k-\I\kappa_j},
\]
and  introduce the matrix
\[
D(k) = \left\{\begin{array}{lll}
\begin{pmatrix} 1 & h^U(k,j)^{-1}\\
-h^U(k,j)& 0 \end{pmatrix}
D_0(k), &  |k-\I\kappa_j|<\delta, & j=1,\dots,N,\\
\begin{pmatrix} 0 & h^L(k,j) \\
-h^L(k,j)^{-1} & 1 \end{pmatrix}
D_0(k), & |k+\I\kappa_j|<\delta, & j=1,\dots,N,\\
\ & \ & \ \\
D_0(k), & \text{else},&
\end{array}\right.
\]
where
\[
D_0(k) = \begin{pmatrix} \Lambda(k)^{-1} & 0 \\ 0 & \Lambda(k) \end{pmatrix}.
\]
Observe that by the property  $\Lambda(-k)=\Lambda^{-1}(k)$ we have
\beq\label{defP}
D(-k)= \sigma_1 D(k) \sigma_1.
\eeq
Now we set
\[
\ti{m}(k)=m(k) D(k).
\]
By \eqref{defP} this conjugation preserves properties III and IV. Moreover
(for details see Lemma~4.2 of \cite{GT}), the jump
corresponding to $ \kappa_0<\kappa_j$ is given by
\beq\label{jumpcondti}
\aligned
\ti{v}(k) &= \begin{pmatrix}1& \frac{\Lambda^2(k)}{h^U(k,j)}
\\ 0 &1\end{pmatrix},
\qquad k\in \T^{j,U}, \\
\ti{v}(k) &= \begin{pmatrix}1& 0 \\ -\frac{1}{
h^L(k,j)\Lambda^2(k)}&1\end{pmatrix},
\qquad k\in\T^{j,L},
\endaligned
\eeq
and the jumps corresponding to $\kappa_0>\kappa_j$ (if any) by
\[
\aligned
\ti{v}(k) &= \begin{pmatrix} 1 & 0 \\ h^U(k,j) \Lambda^{-2}(k)
 & 1 \end{pmatrix},
\qquad  k\in \T^{j,U}, \\
\ti{v}(k) &= \begin{pmatrix} 1 & h^L(k,j) \Lambda^2(k) \\
0 & 1 \end{pmatrix},
\qquad  k\in \T^{j,L}.
\endaligned
\]
In particular, all jumps corresponding to poles, except for possibly one if
$\kappa_j=\kappa_0$, are exponentially close to the identity for $t\to\infty$. In the latter case we will keep the
pole condition for $\kappa_j=\kappa_0$ which now reads
\[
\aligned
\res_{\I\kappa_j} \ti{m}(k) &= \lim_{k\to\I\kappa_j} \ti{m}(k)
\begin{pmatrix} 0 & 0\\ \I\gamma_j \E^{2t\Phi(\I\kappa_j)} \Lambda(\I\kappa_j)^{-2}  & 0 \end{pmatrix},\\
\res_{-\I\kappa_j} \ti{m}(k) &= \lim_{k\to -\I\kappa_j} \ti{m}(k)
\begin{pmatrix} 0 & -\I\gamma_j \E^{2t\Phi(\I\kappa_j)} \Lambda(\I\kappa_j)^{-2} \\ 0 & 0 \end{pmatrix}.
\endaligned
\]
Furthermore, the jump along $\R$ is given by
\beq \label{jumpcond3}
\ti v(k)=\begin{pmatrix}
1-|R(k)|^2 & - \Lambda^2(k)\ol{R(k)} \E^{-2t\Phi(k)} \\
\Lambda^{-2}(k)R(k) \E^{2t\Phi(k)} & 1
\end{pmatrix},\qquad k\in\R.\eeq
The new Riemann--Hilbert problem
\[
\ti m_+(k)=\ti m_-(k)\ti v(k)
\]
for the vector $\ti m$ preserves its asymptotics \eqref{normcond}
as well as the symmetry condition \eqref{symto}. 
It remains to deform the remaining jump along $\R$ into one which is exponentially close to the identity as well.
We choose two contours $\mathcal C^U=\R+\I\varepsilon/2$,  $\mathcal C^L=\R-\I\varepsilon/2$,
where $\varepsilon=\min\{\kappa, \kappa_N - \delta\}$ with $\kappa$  is from \eqref{decay} (see Figure~\ref{fig:sol}).
This choice of $\varepsilon$ guarantees that the reflection coefficient can be continued analytically into the domain
$0<\im k<\varepsilon$, and $\mathcal C^U$ does not cross $\T^{N,U}$. Since by definition
$\ol{R(k)}=R(-k)$, then the function $\ol R$ extends analytically into the domain $-\varepsilon<\im k<0$, and thus up to $\mathcal C^L$.
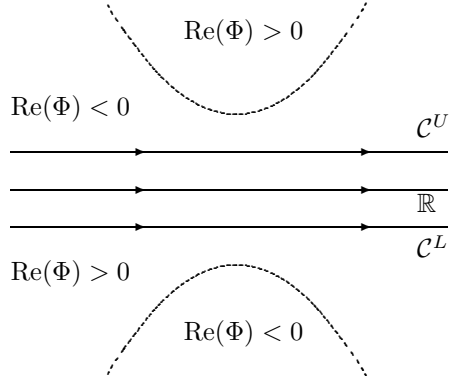
\begin{figure}\centering
\begin{picture}(7,5.2)
\put(0.6,2.5){\line(1,0){5.8}}
\put(2,2.5){\vector(1,0){0.4}}
\put(5,2.5){\vector(1,0){0.4}}

\put(6,2.2){$\R$}

\put(0.6,2){\line(1,0){5.8}}
\put(2,2){\vector(1,0){0.4}}
\put(5,2){\vector(1,0){0.4}}

\put(6,1.6){$\mathcal C^L$}

\put(0.6,3){\line(1,0){5.8}}
\put(2,3){\vector(1,0){0.4}}
\put(5,3){\vector(1,0){0.4}}

\put(6,3.2){$\mathcal C^U$}

\put(0.6,3.5){$\re(\Phi)<0$}
\put(0.6,1.3){$\re(\Phi)>0$}
\put(2.9,4.5){$\re(\Phi)>0$}
\put(2.9,0.5){$\re(\Phi)<0$}

\curvedashes{0.05,0.05}

\curve(1.868,5., 2.1,4.625, 2.6,4., 3.1,3.625, 3.6,3.5, 4.1,3.625, 4.6,4., 5.1,4.625, 5.332,5.)

\curve(1.868,0., 2.1,0.375, 2.6,1., 3.1,1.375, 3.6,1.5, 4.1,1.375, 4.6,1., 5.1,0.375, 5.332,0.)
\end{picture}
\caption{Contour deformation in the soliton region.}\label{fig:sol}
\end{figure}

Now we factorize the jump matrix along $\R$ according to
\[
\tilde v(k)= b_L^{-1}(k) b_U(k)= \begin{pmatrix} 1 & - \Lambda^2(k)R(-k) \E^{-2t\Phi(\I\kappa_{j})}\\
0&1\end{pmatrix}\begin{pmatrix} 1&0\\
\Lambda^{-2}(k)R(k) \E^{2t\Phi(\I\kappa_{j})} & 1
\end{pmatrix}
\]
and set
\beq\label{mbreve}
\breve m(k) =\left\{\begin{array}{ll}\ti m(k) b_U^{-1}(k), & 0<\im k<\varepsilon/2,\\
\ti m(k) b_L^{-1}(k), & -\varepsilon/2<\im k<0,\\
\ti m(k), & \text{else},
\end{array}\right.
\eeq
such that the jump along $\R$ is moved to $\mathcal C^U\cup\mathcal C^L$ and is given by
\[
\breve v(k) =\left\{\begin{array}{ll}
\begin{pmatrix} 1&0\\
\Lambda^{-2}(k)R(k) \E^{2t\Phi(\I\kappa_{j})} & 1
\end{pmatrix}, & k\in\mathcal C^U,\\
\ & \ \\
\begin{pmatrix} 1 & - \Lambda^2(k)R(-k) \E^{-2t\Phi(\I\kappa_{j})} \\
0&1\end{pmatrix}, & k\in\mathcal C^L.
\end{array}\right.
\]
Hence, all jumps $\breve{v}$ are exponentially close to the identity as $t\to\infty$ and one
can use Theorem~A.6 from \cite{KTa} to obtain (repeating literally the proof of Theorem~4.4 in \cite{GT})
the following result:

\begin{theorem}\label{thm:asym}
Assume \eqref{decay1}--\eqref{decay} and abbreviate by $c_j= 4 \kappa_j^2$
the velocity of the $j$'th soliton determined by $\re(\Phi(\I \kappa_j))=0$.
Then the asymptotics in the soliton region, $x/t  \geq \epsilon$ for some small
$\epsilon>0$, are as follows:

Let $\delta > 0$ be sufficiently small such that the intervals
$[c_j-\delta,c_j+\delta]$, $1\le j \le N$, are disjoint and $c_N-\delta>0$.

If $|\frac{x}{t} - c_j|<\delta$ for some $j$, one has
\[
q(x,t) = \frac{-4\kappa_j\gamma_j(x,t)}{(1+(2\kappa_j)^{-1}\gamma_j(x,t))^2} +O(\E^{-\epsilon_4 t}),
\]
where $\min\{\kappa,\kappa_N - \delta\}>\epsilon_4>\varepsilon/2$,
\[
\gamma_j(x,t) = \gamma_j \E^{-2\kappa_j x + 8 \kappa_j^3 t} \prod_{i=j+1}^N \left(\frac{\kappa_i-\kappa_j}{\kappa_i+\kappa_j}\right)^2.
\]

If $|\frac{x}{t} -c_j| \geq \delta$, for all $j$, one has
$
q(x,t) = O(\E^{-\epsilon_4 t})
$.
\end{theorem}

\section{Reduction to the model problem  in the region $-6c^2 t<x<0$ }
\label{secg1}

When the parameter $\xi$ passes through the point $0$ and changes its sign from positive to negative,
the hyperbolas $\re \Phi(k)=0$ start to cross the real axis at the points $k=\pm\sqrt{-\xi}$, $\xi<0$. 
Thus in the holomorphic RH-$k$ problem with the jump matrix $\tilde v(k)$, given by \eqref{jumpcond3} with
\beq\label{newlam}
\Lambda(k,\xi):=\Lambda(k)=\prod_{j=1}^N \frac{k+\I\kappa_j}{k-\I\kappa_j},
\eeq
and \eqref{jumpcondti}, $j=1,\dots,N$, all jumps \eqref{jumpcondti} are exponentially close to the identity matrix for large $t$.
Set
\beq\label{defrla}
\mathcal R(k)=R(k)\Lambda^{-2}(k).
\eeq
This is a continuous function with $|\mathcal R(k)|\neq 0$ for $k\in\R$. Since $\ol{\Lambda(k)}=\Lambda^{-1}(k)$ for $k\in \R$ the matrix $\tilde v(k)$ can be written as
\beq\label{jump16}
\tilde v(k) = \begin{pmatrix}
1-\abs{\mathcal R(k)}^2 & - \ol{\mathcal R(k)} \E^{-2 t \Phi(k)} \\
\mathcal R(k)\E^{2 t \Phi(k)}& 1
\end{pmatrix}, \quad k\in\R.
\eeq
Moreover, by \eqref{RTR}, $\ol {\mathcal R(k)}=\mathcal R^{-1}(k)$ for $k\in [-c,c]$.
We keep the notation $\tilde m(k)$ for the unique solution of the holomorphic RH problem with the jumps \eqref{jump16} and \eqref{jumpcondti} where $\Lambda(k)$ is defined by \eqref{newlam} for $j=1,\dots,N$,
satisfying conditions III--IV of Theorem \ref{thm:vecrhp}.
 
The aim of this section is to reduce the RH problem for $\tilde m(k)$ to a problem with ``almost constant'' jumps, which can be solved explicitly.
To this end we perform a few conjugation and deformation steps. The first one is connected with the so-called $g$-function \cite{dvz}, which replaces the phase function such that the jump matrix
can be factorized in a way which reveals the asymptotic structure. In fact, in the current formulation of the RH problem, the part of the contour from $-\sqrt{-\xi}$ to $\sqrt{-\xi}$ would require a lower/upper
triangular factorization of the jump matrix which is impossible since $|\mathcal R(k)|=1$ for $k\in [-c,c]$. Hence the idea is to perform a conjugation as in Lemma~\ref{lem:conjug} with a function $\tilde d(k)$
such that $\tilde d_+(k)\tilde d_-(k)\E^{-2 t \Phi(k)}=1$ on $[-a,a]$ and $\tilde d_+(k,t)\tilde d^{-1}_-(k,t) = o(1)$ with respect to $t\to\infty$ as $k\in (-a,a)$ for some $a>\sqrt{-\xi}$,  but otherwise the function $g(k)=-\frac{1}{t}\log \tilde d(k) +\Phi(k)$ preserves the qualitative behavior of $\Phi$. This
will lead to a jump matrix
\[
\begin{pmatrix} 0& -\mathcal R(k)\\ \mathcal R(k)& 0\end{pmatrix} +o(1),\quad  k\in[-a,a],
\] 
as $t\to\infty$. A further conjugation step will then turn this into a constant (w.r.t.\ $k$) jump which subsequently has to be solved explicitly.

Set $a=a(\xi)=\sqrt{-2\xi}$. This parameter is positive and monotonous with respect to $\xi$ for $\xi<0$ and covers the interval $(0, c)$ when $\xi$ covers the region under consideration.
In particular, we will use $a>0$ in place of $\xi$ in this section. Explicitly we choose
\beq\label{defggg}
g(k):=g(k,\xi)=4\I (k^2 - a^2)\sqrt{k^2 -a^2},\quad a=\sqrt{-2\xi},
\eeq
defined in the domain $\mathcal D(\xi)= \clos (\C\setminus [-a,\ a])$. We suppose that $\sqrt{k^2 -a^2}$ takes positive values for $k>a$.
By definition $g(-k)=-g(k)$ for $k\in\mathcal D(\xi)$, $g$ has a jump along the interval $[-a,a]$, and $g_+(k)=-g_-(k)>0$ on the  contour $[-a,a]$, taken with orientation from $-a$ to $a$.
The signature table for $\re g$ is shown in Figure~\ref{fig:sigg}.
\begin{figure}\centering
\begin{picture}(7,5.2)
\put(0.6,2.5){\line(1,0){5.8}}

\put(1.5,2.3){$-a$}
\put(2.5,2.15){$-\sqrt{-\xi}$}
\put(4.,2.15){$\sqrt{-\xi}$}
\put(5.5,2.3){$a$}

\put(0.5,3.5){$\re(g)<0$}
\put(0.5,1.3){$\re(g)>0$}
\put(5.5,3.5){$\re(g)<0$}
\put(5.5,1.3){$\re(g)>0$}
\put(2.9,4.5){$\re(g)>0$}
\put(2.9,0.5){$\re(g)<0$}

\curve(1.75,5., 2.2,3.75, 2,2.5)
\curve(1.75,0., 2.2,1.25, 2,2.5)

\curve(5.75,5., 5.3,3.75, 5.5,2.5)
\curve(5.75,0., 5.3,1.25, 5.5,2.5)

\curvedashes{0.05,0.05}

\curve(1.8,5., 2.5,2.5, 1.8,0.)
\curve(5.7,5., 5,2.5, 5.7,0.)

\end{picture}
\caption{Signature table for $\re(g)$ together with the level curve $\re(\Phi)=0$ (dashed).}\label{fig:sigg}
\end{figure}
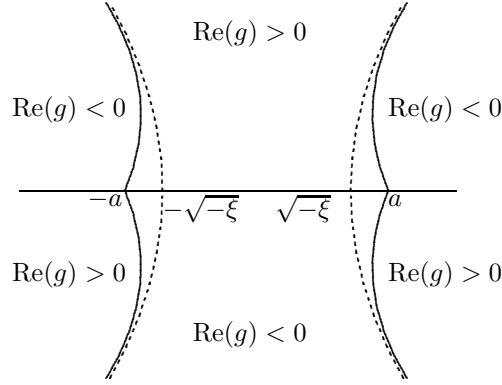
Since
\beq\label{asympg}
\aligned \Phi(k)-g(k)&= 4\I \left(k^3 + 3\xi k - (k^3 + 2\xi k)\sqrt{1 +\frac{2\xi}{k^2}}\right)\\
&=\frac{12\xi^2}{2\I k}(1 +O(k^{-1})), \ \ k\to\infty,\endaligned
\eeq
the function
\[
\tilde d(k):=\E^{t(\Phi(k)-g(k))},\quad k\in \C,
\]
satisfies all conditions of Lemma \ref{lem:conjug}.

STEP 1.  Let $D(k)$ be the matrix \eqref{matrixD} with $d=\tilde d$. Put
$m^{(1)}(k)=\tilde m(k)D(k)$, then $m^{(1)}(k)$ solves the holomorphic RH problem
$m^{(1)}_+(k)=m^{(1)}_-(k)v^{(1)}(k)$ with
\beq\label{jump17}
v^{(1)}(k) = \begin{cases}
\begin{pmatrix}
0 & -\ol {\mathcal R(k)} \\
\mathcal R(k)& \E^{-2tg_+(k)}
\end{pmatrix},\qquad \qquad \ \ k \in [-a,a],\\[4mm]
 \begin{pmatrix}
1-\abs{\mathcal R(k)}^2 & - \ol{\mathcal R(k)} \E^{-2 t g(k)} \\
\mathcal R(k) \E^{2 t g(k)}& 1
\end{pmatrix}, \qquad k\in \R\setminus [-a, a],\\[4mm]
\begin{pmatrix}1&  \tilde h^U(k,j)
\\ 0 &1\end{pmatrix},
\qquad \ \ \ k\in \T^{j,U}, \quad j=1,\dots,N,\\[4mm]
 \begin{pmatrix}1& 0 \\ - \tilde h^L(k,j)&1\end{pmatrix},
\qquad \ \ \ k\in \T^{j,L}, \quad j=1,\dots,N,
\end{cases}
\eeq
where
\[
 \tilde h^U(k,j):=\frac{\Lambda^2(k)}{h^U(k,j)}\E^{2t(\Phi(k) - g(k))},\quad \tilde h^L(k,j):=\frac{1}{\Lambda^{2}(k)h^L(k,j)}\E^{-2t(\Phi(k) - g(k))},
\]
and $h^U(k,j)$, $h^L(k,j)$ are defined by \eqref{remm}.

\begin{lemma}\label{lemest}
Let the radii $\delta$ of the circles $\T^{j,L}$ and $\T^{j,U}$ satisfy the inequalities
\beq\label{valid1}
 (\kappa_N -\delta)^3> 3\delta\left( (\kappa_1+\delta)^2+\frac{c^2}{2}\right)
\eeq
and $\delta<\kappa_N - \kappa$, where $\kappa$ is from \eqref{decay}.  Then, uniformly with respect to $\xi\in [0,-\frac{c^2}{2}]$,
\[
|\tilde h^U(k,j)|+ |\tilde h^L(- k,j)|< C_1(\delta)\E^{-C(\delta) t}, \  k\in \T^{j,U};\ \ C(\delta)>0,\  C_1(\delta)>0.
\]
\end{lemma}
 
\begin{proof}
It is sufficient to check that for sufficiently small $\delta>0$ we have $\re (\Phi(k) - g(k) -\Phi(\I\kappa_j))<0$ when $|k-\I\kappa_j|=\delta$.
The rough estimates, which are valid for $\xi\in (0,c^2/2]$ show that
\[
|\Phi(k)-\Phi(\I\kappa_j)|\leq 12\left( (\kappa_1+\delta)^2+|\xi|\right)\delta\leq 12\delta\left( (\kappa_1+\delta)^2+\frac{c^2}{2}\right),
\]
and $ \re g(k)\geq 4 (\kappa_N -\delta)^3$.  Thus, it is sufficient  to choose $\delta$ satisfying \eqref{valid1}.
\end{proof}

Now set
\[
\T_\delta=\cup_{j=1}^N \left(\T^{j,U}\cup \T^{j,L}\right)
\]
and denote by $\id$ the identity matrix. We observe that the matrix \eqref{jump17} admits the following  representation on $\T_\delta$ :
\beq\label{small1}
v^{(1)}(k,x,t)=\id + A(k,\xi,t),\  \|A(k,\xi,t)\|\leq C_1(\delta)\E^{-C(\delta) t},\  \  C(\delta), C_1(\delta)>0,
\eeq
where $\|A\|=\max _{i,j=1,2} |A_{ij}|$ denotes the matrix norm and the estimate for $A$ is uniform with respect to $k\in\T_\delta$ and $\xi\in[0,-\frac{c^2}{2}]$. 

To perform the next transformation step,  we first consider the following scalar 

\noindent {\bf Conjugation problem}: {\it Find a holomorphic function $d(k)$ in $\C\setminus [-a, a]$ which solves the jump problem
\beq\label{defddr}
d_+(k)d_-(k)=\mathcal R^{-1}(0)\mathcal R(k),\quad k\in [-a,a],
\eeq
and satisfies symmetry and normalization   conditions 
\beq\label{norsym}
d(-k)=d^{-1}(k), \quad k\in \clos\left(\C\setminus [-a, a]\right);\quad d(k)\to 1,\quad k\to \infty.
\eeq
Here $\mathcal R$  is defined by \eqref{newlam} and \eqref{defrla}.}

\begin{lemma}\label{lem:argr}
The function $\arg (\mathcal R(k)\mathcal R^{-1}(0))$ is an odd smooth function on $\R$.
Moreover, $\mathcal R(0)=-1$ in the nonresonant case and $\mathcal R(0)=1$ in the resonant case.
\end{lemma}

\begin{proof}
First of all, recall that $\mathcal R(k)\mathcal R^{-1}(0)$ is continuous and nonzero for $k\in\R$. Therefore its argument is a continuous function. We observe that
\[
\arg \Lambda(k)=\arg\Lambda(0) + G(k) =\pi N + G(k),
\]
where $G(-k)=-G(k)$, $G\in\mathcal C(\R)$. Furthermore, the Levinson theorem (cf.~\cite{Aktosun}, formula (4.3)) yields
\[
\pi N = \frac{\pm\pi Y}{2} + \arg T(0\pm 0),
\]
where $Y=1$ in the nonresonant case, and $Y=0$ in the resonant case. By \eqref{RTR}
\[
\lim_{k\to 0}\arg \mathcal R(k)=\lim_{k\to 0}\left( 2\arg T(k) - 2\arg k - 2\arg \Lambda(k)\right)=-\pi Y.
\]
Thus, the function $\arg(\mathcal R(k)\mathcal R^{-1}(0))$ is a smooth odd function.
Since $\Lambda^2(0)=1$ the value of $\mathcal R(0)$ coincides with the value of the reflection coefficient (see Theorem~\ref{thm:scat}),
that is, $\mathcal R(0)=-1$ in the nonresonant case, and $\mathcal R(0)=1$ in the resonant case.
\end{proof}

To simplify notation introduce
\beq\label{PS}
\mathcal S(k):= \mathcal R(k)\mathcal R^{-1}(0),\quad P(k):= \frac{1}{\sqrt{k^2 -a^2 + \I 0}},\quad k\in [-a, a].
\eeq
To find the solution of the conjugation problem, we transform it to an additive jump problem
\[
f_+(k)=f_-(k) + P(k)\log \mathcal S(k);\quad f(k)\to 0,\quad k\to\infty,
\]
for the function
\[
f(k)=(k^2 - a^2)^{-1/2}\log d(k).
\]
The Sokhotski--Plemelj formula and the property $|\mathcal S|=|\mathcal R|=1$ imply
\beq\label{defef}
f(k)=\frac{1}{2\pi \I}\int_{-a}^a \frac{P(s)\log \mathcal S(s)  }{ \,s-k}ds,
\eeq
where the values of $\log\mathcal S(s) = \I\arg(\mathcal S(s))$ are chosen continuous according to Lemma~\ref{lem:argr}.
Since $\log \mathcal S(s)$ is odd and $P(s)$ is even we note $f(-k)=f(k)$.
Moreover, from the oddness it also follows that
\[
f(k)=\frac{-1}{2\pi\I k} \left(\int_{-a}^a P(s)\log \mathcal S(s) ds +O\left(\frac{1}{k}\right)\right)=O\left(\frac{1}{k^2}\right),\quad k\to\infty.
\]
Thus  $\sqrt{k^2 - a^2}f(k)=O(k^{-1})$ and the function
\beq\label{defsol}
d(k):=\E^{\sqrt{k^2 -a^2}\,f(k)}=\exp\left(\frac{\sqrt{k^2 - a^2}}{2\pi\I}
\int_{-a}^a \frac{\log ( \mathcal R(s)\mathcal R^{-1}(0)) }{\sqrt{s^2 -a^2 +\I 0}\ (s-k)}\,ds\right),
\eeq
satisfies \eqref{defddr} and \eqref{norsym}. Since $f(k)$ is even and $\sqrt{k^2 - a^2}$ is odd, it also satisfies the symmetry condition \eqref{norsym}.
Note also that $d(k)$ is a bounded function in a vicinity of the points $\pm a$ as will be shown in Lemma \ref{lem6} below.

STEP 2. Set $m^{(2)}(k)=m^{(1)}(k)D(k)$ and apply Lemma \ref{lem:conjug} with $d$ given by \eqref{defsol}.
Then we obtain the following RH problem: Find a holomorphic vector function $m^{(2)}(k)$ in the domain $\C\setminus (\R\cup \T_\delta)$,
satisfying conditions III, IV of Theorem \ref{thm:vecrhp} and the jump condition $m^{(2)}_+(k)=m^{(2)}_-(k)v^{(2)}(k)$, where
\[
v^{(2)}(k) = \begin{cases}
\begin{pmatrix}
0 & -\mathcal R(0) \\
\mathcal R(0) & \frac{d_+(k)}{d_-(k)}\E^{-2tg_+(k)}
\end{pmatrix},& k \in [-a,a],\\[4mm]
 \begin{pmatrix}
1-\abs{\mathcal R(k)}^2 & - d(k)^2\ol{\mathcal R(k)} \E^{-2 t g(k)} \\
d(k)^{-2}\mathcal R(k) \E^{2 t g(k)}& 1
\end{pmatrix}, & k\in \R\setminus [-a, a],\\[4mm]
\id +D^{-1}(k)A(k,\xi,t)D(k), & k\in\T_\delta,\end{cases}
\]
with $d(k)$ is given by \eqref{defsol} and $A(k,\xi,t)$ given by \eqref{small1}.

STEP 3. 
The next upper-lower factorization step is standard (cf.\ \cite{dz}, \cite{GT}). Set
\[
v^{(2)}(k) = B^L(k)(B^U(k))^{-1}, \quad k \in \R\setminus [-a,a],
\]
with
\[
B^L(k) =
 \begin{pmatrix}
1 & -d(k)^2\mathcal{R}(-k) \E^{-2 t g(k)} \\
0& 1
\end{pmatrix}, \quad
B^U(k) =  \begin{pmatrix}
1 & 0 \\
-d(k)^{-2}\mathcal R(k) \E^{2 t g(k)}& 1
\end{pmatrix}.
\]
Recall that $\ol{\mathcal R(k)}= \mathcal R(-k)$ for $k\in\R$. This allows us to continue the matrices $B^L(k)$ and $B^U(k)$ to a vicinity of the real axis.
Introduce the domains $\Omega^U$ and $\Omega^L$, bounded by contours $\mathcal C^U$ and $\mathcal C^L$ which are contained in the strip $|\im k|<\kappa/2$,
and asymptotically close to its boundary as $k\to \infty$, as depicted in Figure~\ref{fig:cong}.
Redefine $m^{(2)}$ in $\Omega^U$ and $\Omega^L$ according to
\[
m^{(3)}(k)=\left\{\begin{array}{ll}
  m^{(2)}(k) B^U(k), &k\in \Omega^U,\\[1mm]
 m^{(2)}(k)B^L(k), & k\in \Omega^L,\\[1mm]
 m^{(2)}(k), & \mbox{else.}\end{array}\right.
\]
\begin{figure}\centering
\begin{picture}(7.5,5.2)
\put(0,2.5){\line(1,0){7.5}}
\put(3.3,2.5){\vector(1,0){0.4}}

\put(2.7,3.9){\small $T^{j,U}$}
\put(2.7,2.9){\small$T^{i,U}$}
\put(2.7,1.9){\small$T^{i,L}$}
\put(2.7,0.9){\small$T^{j,L}$}

\put(2.,2.3){$-a$}
\put(5.3,2.3){$a$}

\put(0.5,3.5){$C^U$}
\put(0.5,1.3){$C^L$}
\put(0.5,2.8){$\Omega^U$}
\put(0.5,1.9){$\Omega^L$}
\put(6.5,3.4){$C^U$}
\put(6.5,1.3){$C^L$}
\put(6.5,2.8){$\Omega^U$}
\put(6.5,1.9){$\Omega^L$}

\put(3.6,2.0){\circle*{0.1}}
\put(3.6,2.0){\circle{0.4}}
\put(3.6,3.0){\circle*{0.1}}
\put(3.6,3.0){\circle{0.4}}
\put(3.6,1.0){\circle*{0.1}}
\put(3.6,1.0){\circle{0.4}}
\put(3.6,4.0){\circle*{0.1}}
\put(3.6,4.0){\circle{0.4}}

\put(0.6,3.31){\vector(4,-1){0.2}}
\put(0.6,1.68){\vector(4,1){0.2}}
\put(6.87,3.31){\vector(4,1){0.2}}
\put(6.87,1.68){\vector(4,-1){0.2}}

\curve(0,3.5, 1.5,3, 1.8,2.7, 2.,2.5)
\curve(0,1.5, 1.5,2, 1.8,2.3, 2.,2.5)
\curve(7.5,3.5, 6,3, 5.7,2.7, 5.5,2.5)
\curve(7.5,1.5, 6,2, 5.7,2.3, 5.5,2.5)

\end{picture}
\caption{Contour deformation.}\label{fig:cong}
\end{figure}
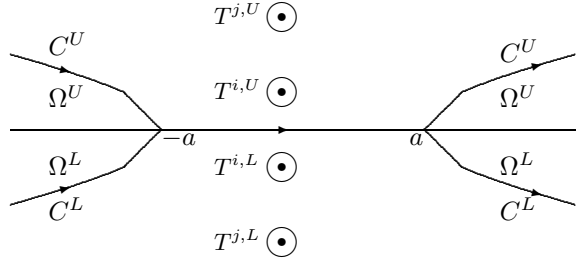%
Then the jumps along the intervals $(-\infty, -a]$ and $[a,\infty)$ disappear and there appear new jumps along $C^U$ and $C^L$ which are asymptotically
close to the identity matrix as $t\to\infty$ away from the points $\pm a$. Moreover, set $A^{(3)}(k)=D^{-1}(k)A(k,\xi,t)D(k)$, $ k\in\T_\delta$, where $D(k)$ is the diagonal matrix associated with \eqref{defsol} and $A$ is from \eqref{small1}.
Then \eqref{small1} and the boundness of $d(k)$ and $d^{-1}(k)$ uniformly on $\T_\delta$ and uniformly with respect to $\xi\in[-c^2/2, 0]$ imply
\beq\label{estA}
\|A^{(3)}(k)\|\leq C\E^{-C t},\  \  C>0,\quad k\in \T_\delta.
\eeq
Moreover, we observe that offdiagonal elements of matrices $B^L(k)$ and $B^U(k)$ are continuous on the contours $\mathcal C^L$ and $\mathcal C^U$ respectively and decay as $k\to \infty$ along the contours exponentially. Indeed, by Lemma \ref{lem6} and \eqref{defggg} we see that $B^U_{21}(k)\to -\mathcal R(0)$ as $k\to \pm a$ and $k\in \mathcal C^U$;
$B^L_{21}(k)\to -\mathcal R(0)$  as $k\to \pm a$ and $k\in \mathcal C^L$; moreover, $v^{(2)}_{22}(k)\to 1$ as $k\to a-0$ and $k\to a+0$, where $k\in\R$. Since contours $\mathcal C^U$ and $\mathcal C^L$ are chosen inside the strip $|\im k|<\kappa$, then by the initial condition $q_0(x)\in C^8(\R)$, the function $\mathcal R(k)=R(k)\Lambda(k)$ behaves as  $\mathcal R(k)=O(k^{-9})$ as $k\to\infty$, $k\in\mathcal C^{U}\cup\mathcal C^{L}$ (cf. \cite{EGLT}). From the other side,  the estimate is valid \beq\label{estge}\exp\{ t g(k)\}=O(\exp\{-2 t |\re k|^{3/2} \kappa\}),\quad k\to\infty,\quad k\in\mathcal C^U,\eeq and by symmetry we get that the offdiagonal elements of $B^U$ and $B^L$ decay exponentially for each $t$ as $k\to\infty$. We proved the following 
\begin{theorem}\label{equivalent}
Let $\xi\in [-c^2/2, 0 ]$. Then the RH problem I--IV (cf.\ Theorem \ref{thm:vecrhp}) is equivalent to the following RH problem:
Find a holomorphic vector function $m^{(3)}(k)$ in $\C\setminus \left(\mathcal C^U\cup\mathcal C^L\cup\T_\delta\cup[-a,a]\right)$,
continuous up to the boundary of the domain, which satisfies:
\begin{enumerate}[(a)]
\item  The jump condition $m_{+}^{(3)}(k)=m_{-}^{(3)}(k) v^{(3)}(k)$, where
\beq \label{jump19}
v^{(3)}(k)= \begin{cases}\begin{pmatrix}
0 & - \mathcal R(0)\\[1mm]
\mathcal R(0) & \frac{d_+(k)}{d_-(k)}\E^{-2tg_+(k)}
\end{pmatrix},& k \in [-a,a],\\[4mm]
\begin{pmatrix}
1 & 0 \\
d(k)^{-2}\mathcal R(k) \E^{2 t g(k)}& 1
\end{pmatrix},& k\in\mathcal C^U,\\[4mm]
\begin{pmatrix}
1 & -d(k)^2\mathcal{R}(-k) \E^{-2 t g(k)} \\
0& 1
\end{pmatrix}& k\in\mathcal C^L,\\[4mm]
\id +A^{(3)}(k), & k\in\T_\delta;\end{cases}\eeq
\item
the symmetry condition \eqref{symto};
\item
the normalization condition \eqref{normcond}.
\end{enumerate}
Here $d(k)$ is defined by \eqref{defsol}, $g(k)$ by \eqref{defggg}, $\mathcal R(k)$ by \eqref{defrla} and \eqref{newlam}, and the matrix $A^{(3)}(k)$ admits the estimate \eqref{estA}.

For $|\im k|>\kappa_1 +1$ the solution $m(k)$ of the initial problem I--IV and the solution of the present problem (a)--(c) are connected via
\beq\label{svyaz}
m^{(3)}(k)=m(k)\begin{pmatrix} h^{-1}(k)&0\\0&h(k)\end{pmatrix},\quad
h(k)=d(k)\Lambda(k)\E^{t(\Phi(k) - g(k))}.
\eeq
\end{theorem}

We observe that the jump matrix $v^{(3)}(k)$ has the structure
\beq\label{struct1}
v^{(3)}(k)=\begin{cases} -\I\mathcal R(0)\sigma_2 +A^{(4)}(k),&k\in [-a,a],\\[1mm]
\id + A^{(5)}(k), &k\in C^U\cup C^L,\\[1mm]
\id + A^{(3)}(k), &k\in\T_\delta,\end{cases}
\eeq
where $\sigma_2$ is the second Pauli matrix and the matrices $A^{(j)}(k)$  admit the estimates
\beq\label{plusj}
\|A^{(j)}(k)\|\leq C\E^{-t \nu(|k^2 - a^2|)},\quad j=4,5.
\eeq
Here $\nu(k)$, $k\in\R_+$, is an increasing positive function as $k\neq 0$  with
$\nu(0)=0$ and $\nu(k)=O(k^{3/4})$ as $k\to +\infty$. 
This structure suggests the shape of a limiting (or model) RH problem, which can be solved explicitly. A solution of this model problem is a contender for the leading term in the asymptotic  expansion for the solution of problem (a)--(c) from Theorem \ref{equivalent} as $t\to\infty$.

\section{The solution of the model problem}\label{model}

In the previous section we were lead to the following model RH problem: 

\noindent{\it Find a holomorphic vector function $m^{\mathrm{mod}}(k)$ in the domain $\C\setminus [-a,a]$, continuous up to the boundary of the domain, except of the endpoints $\pm a$ , where the singularities of the order $O((k\pm a)^{-1/4})$ are admissible, which satisfies the jump condition
\[
m^{\mathrm{mod}}_+(k)=m^{\mathrm{mod}}_-(k)
\begin{pmatrix} 0&-\mathcal R(0)\\ \mathcal R(0)&0\end{pmatrix}, \quad k\in [-a, a];
\]
and the symmetry and normalization conditions:}
\[
m^{\mathrm{mod}}(k)= m^{\mathrm{mod}}(-k)\sigma_1,\qquad m^{\mathrm{mod}}(k)=\rI + O(k^{-1}).
\]
We remark that the solution of this model problem is unique as can be shown using a similar argument as in Theorem~\ref{thm:vecrhp}.
However, this will also follow directly from existence and uniqueness of a solution (to be constructed below) for the associated matrix problem.
Indeed,  two solutions for the vector problem would give two solutions for the matrix problem, violation uniqueness for the matrix problem.

We look for the matrix solution $M^{\mathrm{mod}}(k)=M^{\mathrm{mod}}(k,\xi,t)$ of the {\bf matrix} RH
problem:\\
{\it Find a holomorphic matrix-function $M^{\mathrm{mod}}$ in $\C\setminus [-a,a]$, which has continuous limits to the boundary of the domain, except for the endpoints $\pm a$,
where $M^{\mathrm{mod}}_{ij}=O((k\pm a)^{-1/4})$, $i,j=1,2$, and which satisfies the jump
\[
M_+^{\mathrm{mod}}(k)=-\I \mathcal R(0) M_-^{\mathrm{mod}}(k)\sigma_2,\quad k\in [-a,a],
\] 
and is normalized according to $M^{\mathrm{mod}}(k)=\id +O(k^{-1})$ as $k\to\infty$.}

Note that  $\det (-\I\mathcal R(0)\sigma_2)=1$ and, respectively, $\det M^{\mathrm{mod}}(k)$ is a holomorphic function in $
\C\setminus\{a,-a\}$, with isolated singularities $\det M^{\mathrm{mod}}(k)=O((k\pm a)^{-1/2})$, which are, therefore, removable. By Liouville's theorem and by the normalization condition one has $ \det M^{\mathrm{mod}}(k)=1$. Thus $(M^{\mathrm{mod}})^{-1}(k)=O((k\pm a)^{-1/4})$, $k\to\mp a$, and the rest of the arguments  proving uniqueness are the same as in   \cite[page 189]{deiftbook}. 

The uniqueness and the symmetry $\sigma_1 \sigma_2 \sigma_1 = - \sigma_2$ then imply $M^{\mathrm{mod}}(-k) = \sigma_1 M^{\mathrm{mod}}(k) \sigma_1$.
In turn, the vector solution to our model problem
is given by
\[
m^{\mathrm{mod}}(k)=\rI M^{\mathrm{mod}}(k),
\]
 and hence it fulfills the symmetry condition. 

We construct the solution of the matrix problem following \cite{its}. First consider the resonant case. 
Since
\beq\label{divix}
\sigma_2 = S_0 \sigma_3 S_0^{-1}, \qquad
S_0= \frac{1+\I}{2} \begin{pmatrix} 1& 1\\ \I& -\I\end{pmatrix}, \quad S^{-1}_0= \frac{1-\I}{2} \begin{pmatrix} 1& -\I\\ 1& \I\end{pmatrix}
\eeq
then we can first find a holomorphic solution of the jump problem $M_+^\infty=-\I M_-^\infty \sigma_3$, $M^\infty(\infty)=\id$, where 
$\sigma_3$ is the third Pauli matrix. The solution can be easily computed:
\[
M^\infty(k)=\begin{pmatrix} \beta(k) &0\\ 0 & \beta(k)^{-1}\end{pmatrix},\quad \beta(k)=\sqrt[4]{\frac{k+a}{k-a}}.
\]
Here the function $\beta(k)$ is defined on $\clos(\C\setminus[-a,a])$ and its branch is fixed by the condition $\beta(\infty)=1$. Note that $\beta(-k)=\beta(k)^{-1}$.
For the original matrix function $M^{\mathrm{mod}}(k)$ this yields the representation
\beq\label{solmat}
M^{\mathrm{mod}}(k)= S_0 M^\infty(k) S_0^{-1}= \begin{pmatrix}\frac{\beta(k) +\beta(k)^{-1}}{2} & \frac{\beta(k) - \beta(k)^{-1}}{2\I}\\[2mm]
- \frac{\beta(k) - \beta(k)^{-1}}{2\I} & \frac{\beta(k) +\beta(k)^{-1}}{2}\end{pmatrix}
\eeq
in the resonant case. In the nonresonant case one has to replace $\beta(k)$ by $\beta(-k)$.
The solution of the vector model problem is
\beq\label{modsol}
m^{\mathrm{mod}}(k)=\frac{1}{2\I}\left( \beta(k)(\I-1) +\beta(k)^{-1}(\I+1), \,
\beta(k)(\I+1) +\beta(k)^{-1}(\I-1)\right).
\eeq
In summary we have shown the following
 
\begin{lemma}\label{lembeta}
The solution of the vector (resp.\ the matrix) model RH problems, $m^{\mathrm{mod}}(k)$ (resp.\ $M^{\mathrm{mod}}(k)$) is given by formula \eqref{modsol} (resp.\ \eqref{solmat}), where
$\beta(k)= \sqrt[4]{\frac{k-a}{k+a}}$ in the nonresonant case, and 
$\beta(k)=\sqrt[4]{\frac{k+a}{k-a}}$ in the resonant case.
 \end{lemma}

Before we justify the asymptotic equivalence $m^{(3)}(k)\sim m^{\mathrm{mod}}(k)$ as $t\to\infty$ for $k$ outside of small vicinities of $\pm a$, let us compute what this will imply for the leading asymptotics of the solution of the KdV equation.
By \eqref{svyaz} we have for sufficiently large $k$
\beq\label{as2}
m_1(k)=m^{(3)}_1(k) d(k)\Lambda(k) \E^{t(\Phi(k) - g(k))}\sim m^{\mathrm{mod}}_1(k) d(k)\Lambda(k) \E^{t(\Phi(k) - g(k))}
\eeq
as $t\to\infty$.
By \eqref{asm} we have
\[
q(x,t)=-\frac{\partial }{\partial x}\lim_{k\to\infty} 2\I k\,(m_1(k,\xi,t)-1),
\]
and defining $h(\xi)$ via
\beq\label{as2l}
\Lambda(k)d(k,\xi)m_1^{\mathrm{mod}}(k,\xi)=1 -\frac{h(\xi)}{2\I k}+O\left(\frac{1}{k^2}\right),
\eeq
we have by \eqref{asympg}
\beq\label{as19}
\lim_{k\to\infty} 2\I k\,(m_1(k,\xi,t)-1) \sim 12t \xi^2 - h(\xi).
\eeq 
Thus formally differentiating \eqref{asm} we arrive at
\beq\label{as2f}
 q(x,t)\sim -t \frac{\pa}{\pa x} (12 \xi^2) + h'(\xi) \frac{\partial \xi}{\partial x} = -\frac{x}{6t} + \frac{h'(\xi)}{12t}
\eeq
and hence the leading order comes from the phase alone.

The next two sections are devoted to the proof of this result. In fact, in the following section we will also compute the next term $Q(\xi)$ in the asymptotic expansion
$q(x,t)\sim-\frac{x}{6 t}+ \frac{Q(\xi)}{6 t} + o(t^{-1})$ and show that the only contribution is from \eqref{as2}. So let us also compute this contribution.
Since $\Lambda(k)$ does not depend on $\xi$, it does not affect $Q(\xi)$.
Thus, 
$h^\prime(\xi)$ depends on the respective terms of $d$ and $m_1^{\mathrm{mod}}$ only.
By \eqref{modsol}, in the resonant case
\begin{align*}
m_1^{\mathrm{mod}}(k)&=\frac{1}{2\I} \left(\sqrt[4]{\frac{k+a}{k-a}}(\I-1) +\sqrt[4]{\frac{k-a}{k+a}}(\I+1)\right)\\
&=\frac{1}{2\I}\left( (1+\frac{a}{2k})(\I -1) +(1-\frac{a}{2k})(\I +1)\right) + O(k^{-2})=1-\frac{a}{2\I k} + O(k^{-2}).
\end{align*}
Consequently, in the resonant case
\[
m_1^{\mathrm{mod}}(k)=1+\frac{a}{2\I k} + O(k^{-2}).
\]
Next recall that $P(s)\log\mathcal S(s)$ is an odd function  on the interval $[-a,a]$, where $P$ and $\mathcal S$ are defined by \eqref{PS}.
Then taking into account \eqref{defsol} and $\frac{d}{ds}P^{-1}(s)=s P(s)$ one has
\begin{align*}
d(k) &= \left(1 -\frac{a^2}{2k^2} + O(k^{-4})\right) \exp\left( -\frac{1}{2\pi\I} \int_{-a}^a \frac{P(s)\log \mathcal S(s)}{1-\frac{s}{k}} ds\right)\\
&= 1- \frac{1}{2\pi\I k}\int_{-a}^a s P(s)\log \mathcal S(s)ds +O(k^{-2})\\
&= 1+\frac{1}{2\pi\I k}\int_{-a}^a P^{-1}(s)\,\frac{d}{d s}\log \mathcal R(s)\,ds+O(k^{-2}).
\end{align*}
Thus
\[
h(\xi)= 4 \sum_{j=1}^N \kappa_j \pm a-\frac{1}{\pi}\int_{-a}^a\sqrt{s^2-a^2+\I 0}\,\frac{d}{d s}\log \mathcal R(s)\,ds,
\]
where $\pm$ corresponds to the resonant/nonresonant case, respectively.
Since \[
\frac{\partial a}{\partial x}=\frac{d a}{d \xi}\frac{1}{12 t}=-\frac{1}{12 a t},
\]
then
\[
h'(\xi)=-\frac{1}{a}\left(\pm 1+\frac{a}{\pi}
\int_{-a}^a\frac{\frac{d}{d s}\log \mathcal R(s)}{\sqrt{s^2-a^2+\I 0}} ds\right).
\]
Once \eqref{as2f} is justified this will prove \eqref{coef}.

\section{The parametrix problem}\label{parametrix}

To justify formula \eqref{asq} we study  first the so called parametrix problem, which appears in  vicinities of the node points $\pm a=\pm a(\xi)$.
In these vicinities the jump matrices $A^{(4)}(k)$ and $A^{(5)}(k)$ (cf.\ \eqref{jump19}, \eqref{struct1}, \eqref{plusj}), which were dropped when solving the model problem, are in fact not close to the identity matrix.
The parametrix problem takes their influence into account.

Consider, for example, the point $-a(\xi)$. Let $\mathcal B_-$ be a small open neighborhood of this point.
Abbreviate $\Sigma_1=[-a,\, a]\cap \mathcal B_-$, $\Sigma_2=\mathcal C^U\cap\mathcal B_-$, and $\Sigma_3=\mathcal C^L\cap\mathcal B_-$.
We choose the orientation of these contours as outward from the node point $-a$, that is,
the orientation on $\Sigma_2$ and $\Sigma_3$ is opposite to the orientation on $\mathcal C^U$ and $\mathcal C^L$, respectively.
Inside $\mathcal B_-$ the solution $m^{(3)}$ has jumps only on these contours. 

As a preparation we investigate the behavior of $d(k)$ from \eqref{defsol} as $k\to -a$.

\begin{lemma} \label{lem6}
The following asymptotical behavior is valid as  $k\to -a$:
\beq\label{as11} d(k)^{-2}\mathcal R(k)=\mathcal R(0) +O(\sqrt{k+a}), \   k\notin \Sigma_1;\
\frac{ d_+(k)}{d_-(k)}=1 +O(\sqrt{k+a}),\ k\in\Sigma_1.\eeq
\end{lemma}
\begin{proof} 
To prove  \eqref{as11} we use \eqref{PS}, and represent the integral in \eqref{defsol} as
\beq\label{repr1}
\int_{-a}^a \frac{P(s)\log\mathcal S(s)}{(s-k)}ds=I_1(k) + \I\arg \mathcal S(-a) I_2(k),
\eeq
with
\[
I_1(k)=\int_{-a}^a \frac{ P(s)(\log\mathcal S(s) -\log\mathcal S(-a) )}{s-k}ds,\  I_2(k)=\int_{-a}^a \frac{ P(s)ds}{s-k}.
\]
Since for $a\in (0,c)$ both the reflection coefficient and the Blaschke factor $\Lambda(k)$ are differentiable, we have $\mathcal S(s)-\mathcal S(-a)=O(s+a)$.
Thus
\[
(\log\mathcal S(s) -\log\mathcal S(-a) )(s^2 - a^2)^{-1/2}=O((s+a)^{1/2})
\]
in a vicinity of $-a$. Consequently, (cf.\ \cite{Muskh}, formulas (22.4) and (22.7)) the function 
$I_1(k)$ is H\"older continuous in a vicinity of $-a$ with the finite limiting value
\[
I_1(-a)=\frac{1}{2}\int_{-a}^a \frac{\arg\mathcal S(s) -\arg\mathcal S(-a) }{\sqrt{|a^2 - s^2|}(s-a)}ds
\]
from any direction. The second integral is given by
\[
\frac{1}{2\pi\I} I_2(k)=\frac{1}{2\sqrt{k^2 -a^2}},
\]
as a solution of the jump problem $F_+(k)=F_-(k) + P(k)$, $k\in[-a,a]$; $F(k)\to 0$ as $k\to\infty$.
Substituting this into \eqref{defsol} and taking into account \eqref{repr1} yields
\begin{align}\nn
 \log d(k)&=-\frac{1}{2} \I\arg\mathcal{S}(-a) + \frac{I_1(-a)}{\pi\I}\sqrt{k^2 - a^2}+ O(k+a)\\ \label{posted}
 &=-\frac{1}{2}\log\mathcal S(-a)+\tilde I (-a)\sqrt{k+a} + O(k+a),
\end{align}
where
\beq\label{IA}
 \tilde I(-a)=\frac{\sqrt{2 a}}{2\pi}\int_{-a}^a \frac{\arg\mathcal S(s) -\arg\mathcal S(-a) }{\sqrt{a^2 - s^2}(s+a)}ds.
\eeq
Note that the main term in the representation of $\log d_+(k)$ and $\log d_-(k)$ in the vicinity of $-a$ is evidently the same.
Formula \eqref{posted} then proves \eqref{as11}.
\end{proof}

This lemma allows us  to replace the jump matrix \eqref{jump19} inside $\mathcal B_-$  approximately by the matrix
\beq\label{parini1}
v^{\mathrm{par}}(k):=\E^{-t g_-(k)\sigma_3}\, S\, \E^{t g_+(k)\sigma_3},
\eeq
where 
\beq\label{parini2}S=\begin{cases} S_1:=\begin{pmatrix}  0&-\mathcal R(0)\\ \mathcal R(0)& 1\end{pmatrix}, & k\in\Sigma_1,\\[3mm]
 S_2:=\begin{pmatrix}  1&0\\ -\mathcal R(0)& 1\end{pmatrix}, & k\in\Sigma_2,\\[3mm]
 S_3:= \begin{pmatrix}  1&\mathcal R(0)\\0& 1\end{pmatrix}, & k\in\Sigma_3.\end{cases}
\eeq
Since $\mathcal R(0)^2=1$ we have $S_1 S_2 S_3=\id$ and $\det(S_j)=1$.

We  look for a matrix solution in $\mathcal B_-\setminus(\Sigma_1\cup\Sigma_2\cup \Sigma_3)$ of the jump problem
\beq\label{parini}
M^{\mathrm{par}}_+(k)=M^{\mathrm{par}}_- (k)v^{\mathrm{par}}(k),
\eeq
which is, in some sense to be made precise below, asymptotically close to $M^{\mathrm{mod}}(k)$ on the boundary $\partial\mathcal B_-$ as $t\to\infty$ (cf.\ also \cite{its}).
If $M^{\mathrm{par}}$ solves \eqref{parini}, then the matrix function
\[
M(k)=M^{\mathrm{par}}(k)\E^{-t g(k)\sigma_3}
\]
solves the constant jump problem
\[
M_+(k)= M_-(k) S,
\]
with the normalization $M\sim M^{\mathrm{mod}}\E^{-t g\sigma_3}$ on $\partial\mathcal B_-$.

To simplify our considerations we will next use a change of coordinates which will put the phase into a standardized form
and at the same time rescales the problem. To this end note that in a small vicinity of $-a$ the $g$-function can be represented as
\[
g(k)=8\sqrt{2} a^{3/2} (k+a)^{3/2}(1 + O(k+a)),\ \mbox{as}\ \ k\to -a,
\]
where the branch cut is taken along $[-a, +\infty)$ and the branch is fixed by $(\epsilon +\I 0)^{3/2}>0$ for $\epsilon >0$.
The error term depends only on $a$ and is uniform  on compact sets.
Thus we can introduce a local variable
\beq\label{zamena}
w(k):=\left(\frac{3 t g(k)}{2}\right)^{2/3},
\eeq
for which we have
\beq\label{changev}
w(k)= t^{2/3} C_1 (k+a)(1 + O(k+a)),\quad C_1 = 2\cdot 6^{2/3} a>0,\quad k\to -a,
\eeq
such that $w(k)$ is a holomorphic change of variables. Moreover, we choose the set $\mathcal B_-$ to be the preimage
under the map $k\mapsto w$ of the circle $\mathbb D_\rho$ of radius $t^{2/3} C_1 \rho$, with $\rho<a/4$, centered at $w=0$.  
Furthermore, without loss of generality we can choose the contours $\mathcal C^U$ and $\mathcal C^L$ such that the segments
$\Sigma_1\cup\Sigma_2\cup \Sigma_3$ are mapped onto the straight lines $\left(\Gamma_1\cup\Gamma_2\cup\Gamma_3\right)\cap \mathbb D_\rho$, where
\[
\Gamma_2=\{w\in\C:\,\arg w=\frac{2\pi\I}{3}\},\  \Gamma_3=\{w\in\C:\,\arg w=\frac{4\pi\I}{3}\},\  \Gamma_1=[0,\,+\infty).
\]
Compare Figure~\ref{fig:disc}.
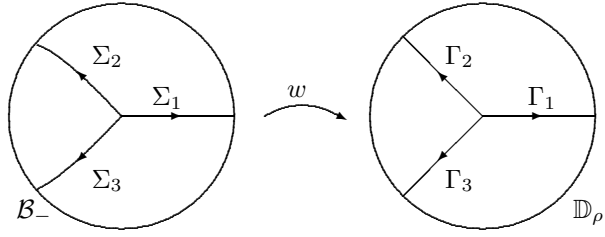
\begin{figure}\centering
\begin{picture}(7.7,3.5)

\put(0, 0.4){$\mathcal B_-$}
\put(1.8,1.9){$\Sigma_1$}
\put(1,2.45){$\Sigma_2$}
\put(1,0.8){$\Sigma_3$}
\put(1.4,1.75){\bigcircle{3}}
\put(1.4,1.75){\line(1,0){1.5}}
\put(2,1.75){\vector(1,0){0.2}}
\curve(0.28,2.7, 0.74,2.4, 1.4,1.75)
\curve(0.28,0.78, 0.74,1.1, 1.4,1.75)
\put(1,2.16){\vector(-1,1){0.2}}
\put(1,1.34){\vector(-1,-1){0.2}}

\put(3.6,2){$w$}
\curve(3.3,1.75, 3.8,1.9, 4.3,1.75)
\put(4.2,1.8){\vector(2,-1){0.2}}

\put(7.4, 0.4){$\mathbb D_\rho$}
\put(6.8,1.9){$\Gamma_1$}
\put(5.7,2.45){$\Gamma_2$}
\put(5.7,0.8){$\Gamma_3$}
\put(6.2,1.75){\bigcircle{3}}
\put(6.2,1.75){\line(1,0){1.5}}
\put(6.8,1.75){\vector(1,0){0.2}}
\put(6.2,1.75){\line(-1,1){1.06}}
\put(5.8,2.15){\vector(-1,1){0.2}}
\put(6.2,1.75){\line(-1,-1){1.06}}
\put(5.8,1.35){\vector(-1,-1){0.2}}

\end{picture}
\caption{The local change of coordinates $w(k)$.}\label{fig:disc}
\end{figure}%
Then the matrix problem \eqref{parini1}--\eqref{parini} can be considered as problem in terms of $w\in \mathbb D_\rho$.

From now on we have to distinguish between the resonant and nonresonant case. Consider first the generic nonresonant case where 
\beq\label{defS}
S_1=\begin{pmatrix} 0& 1\\-1 & 1\end{pmatrix},\ \ S_2=\begin{pmatrix} 1& 0\\1 & 1\end{pmatrix},\ \ S_3=\begin{pmatrix} 1& -1\\0 & 1\end{pmatrix},
\eeq
and the function $\beta(k)$ is locally given by (cf.\ Lemma \ref{lembeta})
\[
\beta(k)= w^{-1/4} \gamma(w), \qquad w\in\mathbb D_\rho,
\]
where $\gamma$ is holomorphic and satisfies
\[
\gamma(w)= 2^{2/3} 3^{1/6} \sqrt{a} \E^{\frac{\I\pi}{4}} t^{1/6} \big(1+ O(\frac{w}{t^{2/3}})\big), \quad\text{as} \quad w\to0,
\]
where the error depends only on $a$ and is uniform on compact sets, which do not contain the point $a=0$.
In turn, \eqref{solmat} can be represented as (cf.\ \eqref{divix})
\beq\label{primo}
 M^{\mathrm{mod}}(k)= S_0 \gamma(w)^{\sigma_3} w^{-\frac{\sigma_3}{4}} S_0^{-1}.
\eeq
Since $\mathbb D_\rho$ grows as $t\to\infty$ this suggests to look for a matrix $\mathcal A(w)$
satisfying the jump condition \beq\label{jump99}\mathcal A_+ = S_j \mathcal A_- \quad \mbox{on}\quad \Gamma_j,\eeq and the normalization
\beq\label{normairy}
\mathcal A(w) = w^{-\frac{\sigma_3}{4}}(S_0^{-1} +O(w^{-3/2}))\E^{-\frac{2}{3} w^{3/2}\sigma_3}, \quad \text{as}\quad w\to\infty,
\eeq
in any direction with respect to $w$. Then
\begin{align}\nn
M^{\mathrm{par}}(k) &= S_0 \gamma(w)^{\sigma_3} \mathcal A(w) \E^{\frac{2}{3} w^{3/2}\sigma_3}\\ \label{solpar}
&= M^{\mathrm{mod}}(k) S_0 \left(\tfrac{3 t g(k)}{2}\right)^{\sigma_3/6} \mathcal A\left(\big(\tfrac{3t g(k)}{2}\big)^{2/3}\right) \E^{t g(k) \sigma_3}
\end{align}
will satisfy
\beq\label{asympar}
M^{\mathrm{par}}(k)= M^{\mathrm{mod}}(k) \big(\id +O(\rho^{-3/2}t^{-1})\big), \quad \text{as}\quad t\to\infty, \quad k\in \partial \mathcal B_-,
\eeq
with the error term again depending only on $a$ and  uniform as $a\in[\epsilon_1, c-\epsilon_2]$ for arbitrary small $\epsilon_j>0$.
The solution the problem \eqref{jump99}, \eqref{normairy} can be given in terms of Airy functions. To this end set
\[
y_1(w)=\Ai(w):=\frac{1}{2\pi \I}\int_{-\I \infty}^{\I \infty} \exp(\frac{1}{3} z^3 - w z) dz,
\]
and let
\[
y_2(w)=\E^{-\frac{2\pi \I}{3}} \Ai(\E^{\frac{-2\pi\I}{3}}w),\ \mbox{and}\ \  y_3(w)=\E^{\frac{2\pi \I}{3}} \Ai(\E^{\frac{2\pi\I}{3}}w).
\]
These functions are {\it entire functions}, and they are connected by the well-known identity \dlmf{9.2.12}
\beq\label{idd}
y_1(w)+y_2(w)+y_3(w)=0.
\eeq
Furthermore, set
\[
\Omega_1=\{w:\ \arg w\in \left(0,\,\frac{2\pi}{3}\right)\},\ \Omega_2=\{
 \arg w\in \left(\frac{2\pi}{3},\,\frac{4\pi}{3}\right)\},\ \Omega_3=\C\setminus\ol{\{\Omega_1\cup\Omega_2\}}.
 \]
 We chose the cuts for all roots of $w$ along the contour $\Gamma_1$ and $\arg w\in [0,\, 2\pi)$. With this convention
 the asymptotics of the Airy functions (cf.\ \dlmf{9.7.5,9.7.6}) read
\begin{align}\label{airas}
y_1(w)&= \begin{cases} \frac{1}{2\sqrt\pi w^{1/4}}\E^{-\frac 2 3 w^{3/2}}(1 +O(w^{-3/2})),& w\in\Omega_1,\\
\frac{\I}{2\sqrt\pi w^{1/4}}\E^{\frac 2 3 w ^{3/2}}(1 +O(w^{-3/2})),& w\in\Omega_3,
\end{cases}\\ \label{airas2}
y_2(w)&=-\frac{\I}{2\sqrt\pi w^{1/4}} \E^{\frac 2 3 w^{3/2}}(1 +O(w^{-3/2})),\quad w\in\Omega_1\cup\Omega_2,\\ \label{airas3}
y_3(w)&=-\frac{1}{2\sqrt\pi w^{1/4}}\E^{-\frac 2 3 w^{3/2}}(1 +O(w^{-3/2})), \quad w\in\Omega_2\cup\Omega_3,
\end{align}
and can be differentiated with respect to $w$. Set
\[
(1-\I)\sqrt\pi\begin{pmatrix}  y_1(w) & y_2(w) \\ - y_1^\prime(w) & - y_2^\prime(w)\end{pmatrix}=: \mathcal{A}_1(w),\quad w\in \Omega_1.
\] 
Then $\det \mathcal A_1(w)=1$ (cf.\ \dlmf{9.2.8}), and by \eqref{airas}, \eqref{airas2} we have the correct normalization \eqref{normairy} in $\Omega_1$.

Next, by \eqref{idd}
\[
\mathcal A_1(w)S_2= (1-\I)\sqrt\pi\begin{pmatrix}-y_3 (w) & y_2(w)\\   y_3^\prime(w) & - y_2^\prime(w)\end{pmatrix}=:\mathcal A_2(w),
\]
and we will use this definition in the sector $\Omega_2$. Again $\det \mathcal A_2(w)=1$ and by
 \eqref{airas2}, \eqref{airas3} the matrix $\mathcal A_2(w)$ obeys the normalization \eqref{normairy} in $\Omega_2$. Finally,
\[
\mathcal A_2 S_3= \mathcal A_1(w)S_1^{-1}= (1-\I) \sqrt\pi\begin{pmatrix} - y_3 (w) & -  y_1(w)\\  y_3^\prime(w) &  y_1^\prime(w)\end{pmatrix}=:\mathcal A_3(w),
\]
has the desired properties in the domain $\Omega_3$. In summary, $\mathcal A(w) = \mathcal A_j(w)$ for $w\in\Omega_j$ is the solution we look for.

\begin{corollary}\label{uhh}
The parametrix $M^{\mathrm{par}}(w)$ defined in \eqref{solpar} satisfies $\det M^{\mathrm{par}}=1$ and is bounded in $ \ol \C$.
\end{corollary}

Taking into account the second term of the Airy functions (cf.\ again \dlmf{9.7.5,9.7.6}), we get from \eqref{asympar} that
\beq\label{import7}
(M^{\mathrm{mod}}(k))^{-1}\,M^{\mathrm{par}}(k) =\id + \frac{1}{72 t g(k)}\begin{pmatrix} -7 & \ 7\\ \ 5&-5\end{pmatrix} +O(t^{-2})
\eeq
uniformly on the boundary $\partial \mathcal B_-$.

Let $\mathcal B_+$ be a vicinity of the point $a$, symmetric to $\mathcal B_-$ with respect to the map $k\mapsto -k$.
Using the symmetry properties of the jump matrices in $\mathcal B_\pm$ and the symmetry of the model problem solution $M^{\mathrm{mod}}(-k)=\sigma_1 M^{\mathrm{mod}}(k)\sigma_1$, one can set 
\[
M^{\mathrm{par}}(k)=\sigma_1 M^{\mathrm{par}}(-k)\sigma_1,\quad k\in\mathcal B_+,
\]
and check directly, that it is indeed the solution of the corresponding parametrix problem in $\mathcal B_+$.  
Note also that since $\det M^{\mathrm{par}}(k)=1$ this matrix is invertible and both $M^{\mathrm{par}}(k)$ and $(M^{\mathrm{par}})^{-1}(k)$ are bounded for all $k\in \clos(\mathcal B_+\cup \mathcal B_-)$ and all $t>0$.

At the end of  his section we  briefly discuss the parametrix problem solution in the resonant case.
The scheme is the same. The $S$ matrix is now given by
\beq\label{defS1}
S_1=\begin{pmatrix} 0& -1\\1 & 1\end{pmatrix},\ \ S_2=\begin{pmatrix} 1& 0\\-1 & 1\end{pmatrix},\ \ S_3=\begin{pmatrix} 1& 1\\0 & 1\end{pmatrix},
\eeq
and $\beta(k)=(2a)^{-1/4}\E^{\frac{-\I\pi}{4}}(k+a)^{1/4}$. We represent the matrix \eqref{solmat} as
\[
M^{\mathrm{mod}}(k)=
S_0 \begin{pmatrix}\beta(k)^{-1}&0\\[2mm]
0&\beta(k)\end{pmatrix} S_0^{-1}.
\]
Thus (cf.\ Lemma \ref{lembeta}),
\[
 M^{\mathrm{mod}}(k)= \ti{S}_0 \gamma(w)^{\sigma_3} w^{-\frac{\sigma_3}{4}} \ti{S}_0^{-1},\quad \gamma(k)=\beta(k(w))^{-1} w^{1/4},
\]
where
 \[
 \ti{S}_0 = \frac{1+\I}{2} \begin{pmatrix} 1 & -1\\ -\I & -\I\end{pmatrix}, \qquad
 \ti{S}_0^{-1} = \frac{1-\I}{2} \begin{pmatrix} -\I & 1\\ \I & 1\end{pmatrix}.
\]
The normalization \eqref{normairy} will have the form
\beq\label{airy1res}
\mathcal{A}(w):=w^{-\frac{\sigma_3}{4}}(\ti{S}_0^{-1} +O(w^{-3/2}))\E^{-\frac{2}{3} w^{3/2}\sigma_3},
\eeq
and 
\[
\mathcal{A}_1(w)= (1-\I) \sqrt\pi\begin{pmatrix}  y_1(w) & - y_2(w) \\  y_1^\prime(w) & - y_2^\prime(w)\end{pmatrix},\quad w\in \Omega_1.
\]

\section{The completion of the asymptotic analysis}\label{asymptotics}

The aim of this section is to establish that the solution $m^{(3)}(k)$ of the RH problem (a)--(c) from
Theorem~\ref{equivalent}, is well approximated by $\rI M^{\mathrm{par}}(k)$ inside the domain $\mathcal B = \mathcal B_+ \cup \mathcal B_-$
and by $\rI M^{\mathrm{mod}}(k)$ in $\C\setminus\mathcal B$.
We follow the well-known approach via singular integral equations (see e.g., \cite{dz}, \cite{GT}, \cite{its}, \cite{len}).
To simplify notations we introduce
\[
\ti\Sigma=[-a,a]\cup\mathcal C^U\cup\mathcal C^L \cup \mathbb T_\delta\cup\partial \mathcal B,\quad
\Sigma_{\pm}=\tilde\Sigma\cap \mathcal B_\pm, \quad
\Sigma_\mathcal B=\tilde\Sigma\cap \mathcal B.
\]
We will denote the three parts of each contour $\Sigma_+$ and $\Sigma_-$, with the orientation as on $[-a,a]\cup\mathcal C^U\cup \mathcal C^L$,
by $\Sigma_j^+$ and $\Sigma_j^-$. Next set
\beq\label{dehatm}
\hat m(k)=m^{(3)}(k) (M^{\text{as}}(k))^{-1},\quad M^{\text{as}}(k):=\begin{cases} M^{\mathrm{par}}(k), & k\in\mathcal B,\\
 M^{\mathrm{mod}}(k), & k\in\C\setminus\mathcal B.
\end{cases}
\eeq
Then $\hat m$ solves the jump problem
\[
\hat m_+(k)=\hat m_-(k)\hat v(k),
\]
where
\beq\label{hatve}
\hat v(k)=\begin{cases}M^{\mathrm{par}}_-(k)v^{(3)}(k)(M^{\mathrm{par}}_+(k))^{-1}, & k\in\Sigma_{\mathcal B},\\ 
(M^{\mathrm{mod}}(k))^{-1}M^{\mathrm{par}}(k), & k\in \partial \mathcal B,\\
M_-^{\mathrm{mod}}(k)v^{(3)}(k)(M_+^{\mathrm{mod}}(k))^{-1}, & k\in \ti\Sigma\setminus(\Sigma_{\mathcal B} \cup \partial \mathcal B),\end{cases}
\eeq
and satisfies the symmetry and the normalization conditions:
\beq\label{hatnorm}
\hat m(k)=\hat m(-k)\sigma_1,\qquad  \hat m\to \rI, \ \ k\to\infty.
\eeq
Abbreviate $W(k)=\hat v(k) -\id$. Then 
\begin{equation}\label{Wi}
W(k)= \begin{cases}M^{\mathrm{par}}_-(k)\left(v^{(3)}(k)-v^{\mathrm{par}}(k)\right)(M^{\mathrm{par}}_+(k))^{-1}, & k\in\Sigma_{\mathcal B},\\ 
(M^{\mathrm{mod}}(k))^{-1}\,M^{\mathrm{par}}(k) -\id, & k\in \partial\mathcal B,\\
M_-^{\mathrm{mod}}(k)(v^{(3)}(k) +\I\mathcal R(0)\sigma_2)(M_+^{\mathrm{mod}}(k))^{-1}, & k\in [-a,a]\setminus \Sigma_{\mathcal B},\\
M_-^{\mathrm{mod}}(k)(v^{(3)}(k)-\id)(M_+^{\mathrm{mod}}(k))^{-1}, &  k\in \ti\Sigma\setminus(\Sigma_{\mathcal B} \cup \partial \mathcal B\cup [-a,a]).\end{cases}
\end{equation}
By construction the function $W(k)$ depends smoothly on $\xi$ when $\xi\in \mathcal I=[-\frac{c^2}{2} +\epsilon, -\epsilon]$, for arbitrary small fixed
positive $\epsilon$. Since $a(\xi)>\sqrt{2\epsilon}$ we assume  that  the minimal  radius $\rho$ of the sets $\mathcal B_\pm$ admits the estimate $\rho\geq \frac{1}{4}\sqrt{2\epsilon}$.

First we study $W(k)$ on $\Sigma_{\mathcal B}$. The matrices $M^{\mathrm{par}}_-(k)$ and $(M^{\mathrm{par}}_+(k))^{-1} $ are smooth bounded functions with respect to $k\in\Sigma_{\mathcal B}$, $t\in [1,\infty)$, and $\xi\in \mathcal I$. The matrix $v^{(3)}(k)-v^{\mathrm{par}}(k)$ has one nonvanishing entry on each part of contour $\Sigma_{\mathcal B}$, which we denote by $u_\pm(k)$:
\[
u_\pm(k)=\begin{cases}(\mathcal R(0)- d(k)^2\mathcal R(-k))\E^{-2tg(k)}, & k\in\Sigma_3^\pm,\\
(d(k)^{-2}\mathcal R(k) - \mathcal R(0))\E^{2tg(k)}, & k\in\Sigma_2^\pm,\\           
(\frac{d_+(k)}{d_-(k)}-1)\E^{-2tg_+(k)}, & k\in\Sigma_1^\pm.           \end{cases}
\]
Since $g(k)=\re g(k)$ on $\Sigma_\pm$, then by \eqref{as11}, \eqref{posted} and \eqref{IA}
\beq\label{main20}
u_\pm(k)=\left( C_j^\pm\tilde I(\pm a)\sqrt{|k\mp a|} \right) \E^{-2(2a)^{3/2} \,t\,|k\mp a|^{3/2}}+ O(k\mp a)\E^{-2tg(k)},\ k\in\Sigma_j^\pm,
\eeq
where $(C_j^\pm)^6 =1$. First of all, we observe that
\beq\label{uest}
u_\pm(k)=O(t^{-1/3}),\quad k\in\Sigma_\pm,
\eeq
where the error $O(t^{-1/3})$ is uniformly bounded with respect to $\rho=\rho(\xi)$ and $a=a(\xi)$ for $\xi\in\mathcal I$. 
Moreover, in this section, the notation $O(t^{-\ell})$ will always denote a function of $a,\rho$ and $t$ with the above mentioned properties.
It is defined for $t\in [T_0,\infty)$, where $T_0=T_0(\epsilon)$ is some large positive time. 

Now let $(\pm a + (C_j^\pm)^2\delta_j^\pm)$ be the end points of the contours $\Sigma_j^\pm$. Recall that $\delta_j^\pm\geq\rho\geq\frac{\sqrt{2\epsilon}}{4}$.
Then
\beq \label{u1est}
\int_{\Sigma_j^\pm} u_\pm(k)dk= C_j^\pm\tilde I(\pm a) \int_0^{\delta_j^\pm} y^{1/2} \E^{-8 t a\sqrt{2a}y^{3/2}}dy +O(t^{-4/3})=\frac{F_\pm(a,j)}{t}+O(t^{-4/3}),
\eeq
where $F_\pm(a,j)=C_j^\pm\tilde I(\pm a)(12 a\sqrt{2a})^{-1}$, and
\beq\label{u2est}
\|u_\pm(k)\|_{L^1(\Sigma_\pm)}=O(t^{-1}).
\eeq 
Moreover, using the same arguments taking into account that the matrix entries
$[ M_-^{\text{par}}]_{rs}(k)[( M^{\text{par}}_+)^{-1}]_{pq}(k)$, $r,s,p,q\in\{1,2\}$, are bounded for $k\in\Sigma_{\mathcal B}$,
uniformly with respect to $\xi\in\mathcal I$, and using \eqref{changev}, \eqref{main20} and Corollary \ref{uhh}, we get for $\ell=0,1$:
\beq\label{est78}
\sum_\pm \int_{\Sigma_\pm} k^\ell u_\pm(k)[M_-^{\text{par}}]_{rs}(k)[(M^{\text{par}}_+)^{-1}]_{pq}(k)dk 
   =\frac{h_{p,q,r,s,\ell}(a)}{t} +O(t^{-4/3}).
\eeq
Here the functions $h_{p,q,r,s,\ell}(a)$ are  bounded with respect to $\xi\in\mathcal I$ and the estimate \eqref{est78} implies that
\beq\label{uses}
\int_{\Sigma_\mathcal B} k^\ell \,W(k)dk= \frac{F_{2,\ell}(a)}{t} + O(t^{-4/3}),\quad \ell=0,1,
\eeq
where the matrices $F_{2,\ell}(a)$ are bounded for $\xi\in\mathcal I$. We also have
\beq\label{norW}
\|k^\ell\,W(k)\|_{L^1(\Sigma_{\mathcal B})}=O(t^{-1}), \quad  \|k^\ell\,W(k)\|_{L^\infty(\Sigma_{\mathcal B})}=O(t^{-1/3}).
\eeq
Moreover, from \eqref{Wi} and \eqref{import7} it follows that
\beq\label{uses2}
 \int_{\partial \mathcal B} k^\ell\,W(k)dk =
\frac{F_{3,\ell}(a)}{t\,\rho^{1/2}} + O(t^{-4/3}),
\eeq
where the matrices $F_{3,\ell}(a)$ have the same properties as $F_{2,\ell}(a)$. Next, the matrix $M^{\text{mod}}(k)$ and its inverse are bounded with an estimate $O(\rho^{-1/4})$
on  the remaining part of the contour $\tilde\Sigma$. Using \eqref{estge}, \eqref{Wi}, \eqref{struct1}, \eqref{plusj}, and \eqref{estA}
we conclude that for $ \ell=0,1$:
\beq\label{laste}
\int_{\tilde\Sigma\setminus(\Sigma_{\mathcal B}\cup \partial \mathcal B)} k^\ell\,W(k) dk=\tilde F_\ell(a,\rho, t),\quad \|\tilde F_\ell(a,\rho, t)\|\leq C(\ell)\rho^{-1/4}\E^{-\frac{\rho t}{2}},
\eeq
where the matrix norms of $F_\ell(a,\rho,t)$  are uniformly bounded with respect to $a$ and $\rho$ for $t\in[T_0,\infty)$ and $\xi\in\mathcal I$.
From  \eqref{Wi}, \eqref{struct1}, \eqref{plusj}, and \eqref{estA} it follows also
\[
\|k^\ell W(k)\|_{L^1(\tilde\Sigma\setminus(\Sigma_{\mathcal B}\cup \partial \mathcal B))}\leq O(\E^{-\epsilon t}),\quad \|k^\ell W(k)\|_{L^\infty(\tilde\Sigma\setminus(\Sigma_{\mathcal B}\cup \partial \mathcal B))}\leq O(\E^{-\epsilon t}).
\] 

As a consequence of these considerations (and using interpolation) we get:

\begin{lemma}\label{lemimp}
The following estimates hold uniformly with respect to $\xi\in\mathcal I$:
\begin{equation}
\label{2w}\|W\|_{L^p(\tilde\Sigma)} = O\left(t^{-\frac{1}{3}-\frac{2}{3p}}\right), \qquad 1\le p \le \infty.
\end{equation}
Moreover,
\beq\label{dubl}
\frac{1}{\pi \I^\ell} \rI \int_{\tilde\Sigma} k^\ell W(k) dk = \begin{pmatrix} 1 & (-1)^\ell \end{pmatrix} \frac{f_\ell(a,\rho)}{t} + O\left(t^{-4/3}\right),\qquad j=0,1,
\eeq
where the functions $f_\ell(a,\rho)$ are bounded with respect to $a$ and $\rho$ for $\xi\in\mathcal I$.
\end{lemma}

Now we are ready to apply the technique of singular integral equations. Since this is well known (see, for example, \cite{dz}, \cite{GT}, \cite{len})
we will be brief and only list the necessary notions and estimates. 

Let $\mathfrak C$ denote the Cauchy operator associated with $\ti\Sigma$:
\[
(\mathfrak C h)(k)=\frac{1}{2\pi\I}\int_{\ti\Sigma}h(s)\frac{ds}{s-k}, \qquad k\in\C\setminus\ti\Sigma,
\]
where $h= \begin{pmatrix} h_1 & h_2 \end{pmatrix}\in L^2(\ti\Sigma)\cup L^\infty(\ti\Sigma)$. 
Let  $\mathfrak C_+ f$ and $\mathfrak C_- f$ be its non-tangential limiting values from the left and right sides of $\ti\Sigma$, respectively.
These operators will be bounded with bound depending on the contour, that is on $a$. However, since we can choose our contour
scaling invariant at least locally, scaling invariance of the Cauchy kernel implies that we can get a bound which is uniform on compact sets.

As usual, we introduce the operator $\mathfrak C_{W}:L^2(\ti\Sigma)\cup L^\infty(\ti\Sigma)\to
L^2(\ti\Sigma)$ by $\mathfrak C_{W} f=\mathfrak C_-(f W)$, where $W$ is our error matrix \eqref{Wi}. 
Then,
\[
\|\mathfrak C_{W}\|_{L^2(\ti\Sigma)\to L^2(\ti\Sigma)}\leq C\|W\|_{L^\infty(\ti\Sigma)}\leq O(t^{-1/3})
\] 
as well as
\beq\label{6w}
\|(\id - \mathfrak C_{W})^{-1}\|_{L^2(\ti\Sigma)\to L^2(\ti\Sigma)}\leq \frac{1}{1-O(t^{-1/3})}
\eeq
for sufficiently large $t$. Consequently, for $t\gg 1$, we may define a vector function
\[
\mu(k) =\rI + (\id - \mathfrak C_{W})^{-1}\mathfrak C_{W}\big(\rI\big)(k).
\]
Then by \eqref{2w} and \eqref{6w}
\begin{align}\nn
\|\mu(k) - \rI\|_{L^2(\ti\Sigma)} &\leq \|(\id - \mathfrak C_{W})^{-1}\|_{L^2(\ti\Sigma)\to L^2(\ti\Sigma)} \|\mathfrak C_{-}\|_{L^2(\ti\Sigma)\to L^2(\ti\Sigma)} \|W\|_{L^2(\ti\Sigma)}\\
&= O(t^{-2/3}).\label{estmu}
\end{align}
With the help of $\mu$ the solution of the RH problem \eqref{hatve}--\eqref{hatnorm} can be represented as 
\[
\hat m(k)=\rI +\frac{1}{2\pi\I}\int_{\ti\Sigma}\frac{\mu(s) W(s)ds}{s-k}
\]
and by virtue of \eqref{estmu} and Lemma \ref{lemimp} we obtain as $k\to +\I\infty:$
\beq\label{hatem2}
\hat{m}(k) = \rI - \frac{1}{2\pi\I } \int_{\ti\Sigma} \frac{\rI W(s)}{k-s} ds + H(k),
\eeq
where 
\beq\label{estah}
|H(k)|\leq \frac{1}{\im k}\|W\|_{L^2(\ti\Sigma)}\|\mu (k)- \rI\|_{L^2(\ti\Sigma)}\leq \frac{O(t^{-4/3})}{\im k},
\eeq 
where $O(t^{-4/3})$ is uniformly bounded with respect to $a$ and $\rho$ as $\xi\in \mathcal I$. In the regime $\re k=0, \im k \to +\infty$ we have
\begin{align*}
\frac{1}{2\pi\I } \int_{\ti\Sigma} \frac{\rI W(s)}{k-s} ds &=  \frac{f_0(a,\rho)}{2\I k t} \begin{pmatrix} 1 & -1\end{pmatrix} + \frac{f_1(a,\rho)}{2 k^2 t} \rI\\
& + O(t^{-1})O(k^{-3}) +O(t^{-4/3})O(k^{-1}),
\end{align*}
where $O(k^{-s})$ are vector-functions depending on $k$ only and $O(t^{-s})$ are as above.
From now we can choose $\rho=\sqrt{\frac{\epsilon}{8}}$ and denote $f_\ell(a,\rho):=f_\ell(\xi)$. These functions are  bounded as $\xi\in\mathcal I$, and in fact they are differentiable with respect to $\xi$, but we will not use their smoothness.
By \eqref{dehatm} and \eqref{as2} for large $k\to +\I \infty$ we have
\[
m(k)=\hat m(k)M^{\text{mod}}(k)\left(d(k)\Lambda(k)\E^{t(\Phi(k)-g(k)}\right)^{\sigma_3},
\]
and from \eqref{as2l}, \eqref{as19} , \eqref{hatem2}, \eqref{dubl}, \eqref{solmat}, \eqref{asm} and \eqref{eq:asygf} it follows:
\begin{align*}
\int_x^\infty q(y,t) dy & = 12 t \xi^2 - h(\xi) + \frac{f_0(\xi)}{t} + O(t^{-4/3}),\\
q(x,t) & = -2 \xi - \frac{f_1(\xi)\mp 2a f_0(\xi)}{t} + O(t^{-4/3}).
\end{align*}
In particular, this shows that the first asymptotic formula can be differentiated with respect to $x$ giving
\[
q(x,t) = -2 \xi + \frac{1}{12t} h'(\xi) + O(t^{-4/3}).
\]
This establishes \eqref{as2f} and completes the proof of Theorem \ref{maintheor} {\bf A}.

\section{Asymptotics in the domain $x<-6c^2 t$}\label{sec:left}

Here we solve the  RH$_1$ problem, considered in Theorem \ref{thm:vecrhpn}, and prove claim {\bf B} of Theorem \ref{maintheor}.
Let $k_1^\pm=\pm \sqrt{-\frac{c^2}{2}-\xi}$ be the stationary phase points of the phase function $\Phi_1(k_1)$.
The signature table for $\re \Phi_1$ in the present domain $\xi<-\frac{c^2}{2}$ is shown in Figure~\ref{fig44}.
It shows that in the domain under consideration, the jump matrix $v(k_1)$ is exponentially close to the identity matrix as $t\to\infty$ except for $k_1\in\R$.
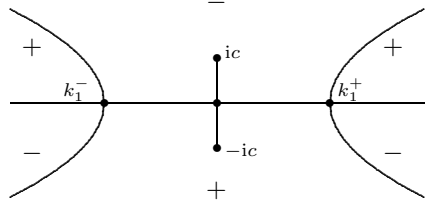
\begin{figure}[h]
\begin{picture}(6,3)
\put(0,1.25){\line(1,0){5.5}}

\put(0.15,0.5){$-$}
\put(0.15,1.9){$+$}
\put(4.95,0.5){$-$}
\put(4.95,1.9){$+$}
\put(2.6,2.5){$-$}
\put(2.6,0.){$+$}

\put(0.7,1.35){$\scriptstyle  k_{1}^-$}
\put(4.35,1.35){$\scriptstyle k_{1}^+$}

\put(1.25,1.25){\circle*{0.1}}
\put(4.25,1.25){\circle*{0.1}}

\put(2.75,1.25){\circle*{0.1}}
\put(2.75,0.65){\line(0,1){1.2}}
\put(2.75,0.65){\circle*{0.1}}
\put(2.75,1.85){\circle*{0.1}}
\put(2.85,1.85){$\scriptstyle  \I c$}
\put(2.85,0.55){$\scriptstyle  -\I c$}

\curve(0.,0., 1.25,1.25, 0.,2.5)

\curve(5.5,0., 4.25,1.25, 5.5,2.5)
\end{picture}
  \caption{Sign of $\re(\Phi_1(k_1))$}\label{fig44}
\end{figure}
Now, following the usual procedure \cite{dz}, \cite{GT}. We let $d^{(1)}(k_1)$ be an analytic function in the domain $ \C\setminus\left(\R\setminus[k_{1}^-,k_{1}^+]\right)$ satisfying
\[
 d_+^{(1)}(k_1)= d_-^{(1)}(k_1)(1-|R_1(k_1)|^2)\ \mbox{for}\ k_1\in\R\setminus[k_{1}^-,k_{1}^+] \ \mbox { and}\ d^{(1)}(k_1)\to 1, \ k_1\to\infty.
\]
By the Sokhotski--Plemelj formulas this function is explicitly given by
\beq\label{defdel}
 d^{(1)}(k_1)=\exp\left(\frac{1}{2\pi \I} \int_{(-\infty,k_1^-) \cup(k_1^+,\infty)} \frac{\log(1-|R_1(s)|^2)}{s-k_1} ds\right).
\eeq
Note that this integral is well defined since $R_1(k_1)=O(k_1^{-1})$ and $|R_1(k_1)|<1$ for $k_1\neq 0$ (cf.\ \cite{EGLT}).
As the domain of integration is even and the function $\log(1-|R_1|^2)$ is also even, we obtain $ d^{(1)}(-k_1)=d^{(1)}(k_1)^{-1}$ and the matrix
\[
 D(k_1)=\begin{pmatrix} d^{(1)}(k_1)^{-1}&0\\ 0& d^{(1)}(k_1)\end{pmatrix}
\]
satisfies the symmetry conditions of Lemma~\ref{lem:conjug}. Now set $ m^{(2)}(k_1)=m^{(1)}(k_1) D(k_1)$, then the new RH$_1$ problem will read
$ m_+^{(2)}(k_1)= m_-^{(2)}(k_1)  v^{(2)}(k_1)$, where

\noindent
$ m^{(2)}(k_1)\to \rI$ as $k_1\to\infty$, $ m^{(2)}(-k_1)= m^{(2)}(k)\sigma_1$, and
\[
  v^{(2)}(k)=\left\{\begin{array}{ll}
 A_L(k_1)A_U(k_1), & k_1\in\R\setminus[k_{1}^-,k_{1}^+],\\
 \ & \\
 B_L(k_1)B_U(k_1), & k_1\in [k_{1}^-,k_{1}^+],\\
 \ & \\
D^{-1}(k_1)v^{(1)}(k_1) D(k_1), & k_1\in [\I c, -\I c]\cup_j(\T_j^U\cup\T_j^L),
\end{array}\right.
\]
where $v^{(1)}(k_1)$ is defined by \eqref{vleft},
\[
A_L(k_1):=\begin{pmatrix}1&0\\ \frac{R_1(k_1) \E^{t\Phi_1(k_1)}}{(1-|R_1(k_1)|^2) d^{(1)}(k_1)^2} &1\end{pmatrix},\quad k\in\Omega_l^L\cup\Omega_r^L,
\]
\[
A_U(k_1):=\begin{pmatrix} 1& -\frac{d^{(1)}(k_1)^2 R_1(-k_1) \E^{-t\Phi_1(k_1)}}{(1-|R_1(k_1)|^2)}\\0&1\end{pmatrix}, \quad k_1\in \Omega_l^U\cup\Omega_r^U,
\]
\[
B_L(k_1):=\begin{pmatrix}1 & -d^{(1)}(k_1)^2 R_1(-k_1)\E^{-t\Phi_1(k_1)}\\ 0 & 1\end{pmatrix},\quad k_1\in \Omega_c^L,
\]
\[
B_U(k_1):=\begin{pmatrix} 1& 0\\
d^{(1)}(k_1)^{-2} R_1(k_1)\E^{t\Phi_1(k_1)}& 1\end{pmatrix}, \quad k_1\in\Omega_c^U.
\]
Here the domains $\Omega_l^L$, $\Omega_l^U$, $\Omega_r^L$, $\Omega_r^U$, $\Omega_c^L$, and $\Omega_c^U$
together with their boundaries $\mathcal C_l^L$, $\mathcal C_l^U$,  $\mathcal C_r^L$, $\mathcal C_r^U$, $\mathcal C_c^L$, and $\mathcal C_c^U$ are shown in Figure~\ref{fig6}.
\begin{figure}[h]
\begin{picture}(7,5.2)
\put(0,2.5){\line(1,0){7.0}}
\put(1.8,2.5){\vector(1,0){0.4}}
\put(5,2.5){\vector(1,0){0.4}}

\put(6.6,2.2){$\R$}

\put(0,2){\line(1,0){2.0}}
\put(5,2){\line(1,0){2.0}}
\put(1.1,2){\vector(1,0){0.4}}
\put(5.5,2){\vector(1,0){0.4}}
\put(2.6,2.9){\vector(2,3){0.2}}
\put(4.12,3.3){\vector(2,-3){0.2}}

\put(1.3,1.5){$\mathcal C_l^L$}
\put(1.3,2.65){$\scriptstyle\Omega_l^U$}
\put(5.7,3,3){$\mathcal C_r^U$}
\put(5.7,2.65){$\scriptstyle\Omega_r^U$}
\put(2.7,3.7){$\mathcal C_c^U$}
\put(2.7,2.65){$\scriptstyle\Omega_c^U$}
\put(3.9,3.7){$\mathcal C_c^U$}
\put(3.8,2.65){$\scriptstyle\Omega_c^U$}

\curve(2.,2., 2.2,2.1, 2.4,2.4, 2.45,2.5)
\curve(4.55,2.5, 4.6,2.6, 4.8,2.9, 5.,3.)
\curve(2.45,2.5, 3.5,1.25, 4.55,2.5)

\put(3.5,1.4){\line(0,1){2.2}}
\put(3.5,3.5){\vector(0,-1){0.4}}
\put(3.5,2){\vector(0,-1){0.4}}

\put(3.68,0.33){\vector(0,-1){0.1}}
\put(3.68,0.78){\vector(0,-1){0.1}}
\put(3.68,4.47){\vector(0,1){0.1}}
\put(3.68,4.93){\vector(0,1){0.1}}

\put(3.5,0.3){\circle*{0.1}}
\put(3.5,0.3){\circle{0.4}}
\put(3.5,0.75){\circle*{0.1}}
\put(3.5,0.75){\circle{0.4}}
\put(3.5,4.5){\circle*{0.1}}
\put(3.5,4.5){\circle{0.4}}
\put(3.5,4.95){\circle*{0.1}}
\put(3.5,4.95){\circle{0.4}}

\put(2.6,0.2){$\scriptstyle T_1^{j,L}$}
\put(2.6,0.65){$\scriptstyle T_1^{i,L}$}
\put(2.6,4.4){$\scriptstyle T_1^{i,U}$}
\put(2.6,4.85){$\scriptstyle T_1^{j,U}$}

\put(3.3,1.0){$\scriptstyle -\I c$}
\put(3.45,4){$\scriptstyle \I c$}

\put(0,3){\line(1,0){2.0}}
\put(5,3){\line(1,0){2.0}}
\put(1.1,3){\vector(1,0){0.4}}
\put(5.5,3){\vector(1,0){0.4}}
\put(2.61,2.1){\vector(2,-3){0.2}}
\put(4.12,1.7){\vector(2,3){0.2}}

\curve(2.,3., 2.2,2.9, 2.4,2.6, 2.45,2.5)
\curve(4.55,2.5, 4.6,2.4, 4.8,2.1, 5,2.)
\curve(2.45,2.5, 3.5,3.75, 4.55,2.5)

\put(1.3,3.3){$\mathcal C_l^U$}
\put(1.3,2.15){$\scriptstyle\Omega_l^L$}
\put(5.7,1.5){$\mathcal C_r^L$}
\put(5.7,2.15){$\scriptstyle\Omega_r^L$}
\put(2.6,1.2){$\mathcal C_c^L$}
\put(2.7,2.15){$\scriptstyle\Omega_c^L$}
\put(4,1.2){$\mathcal C_c^L$}
\put(3.8,2.15){$\scriptstyle\Omega_c^L$}

\curvedashes{0.05,0.05}
\curve(0.,0.05, 0.85,0.5, 1.55,1., 2.05,1.5, 2.45,2.5, 2.05,3.5, 1.55,4., 0.85,4.5, 0.,4.94)
\curve(7.,0.05, 6.15,0.5, 5.45,1., 4.95,1.5, 4.55,2.5, 4.95,3.5, 5.45,4., 6.15,4.5, 7.,4.94)
\end{picture}
\caption{Contour deformation in the domain $x<-6c^2t$}\label{fig6}
\end{figure}
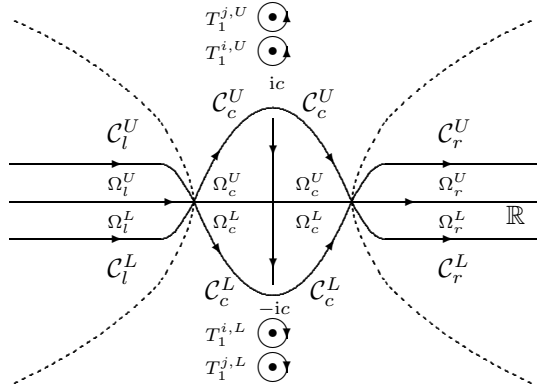
Evidently,  the matrix $ B_U$ (resp. $B_L$) has a jump along the contour $[\I c, 0]$ (resp.\ $[0, -\I c]$).
All contours are oriented from left to right. They are chosen to respect the symmetry $k_1\mapsto -k_1$ and are inside a set, where $R_1(k_1)$ has an analytic continuation.
We also used the analytic continuation $\ol{R_1(k_1)}=R_1(-k_1)$ to these domains.

\begin{lemma}\label{plucker}
The following formula is valid
\[
(B_{U})_-\, v^{(2)} \, (B_{U})_+^{-1}=\id,\quad k_1\in [\I c, 0];\quad (B_{L})_-^{-1}\, v^{(2)} \,(B_{L})_+=\id,\quad k_1\in [0, -\I c].
\]
\end{lemma}

\begin{proof}
By virtue of the Pl\"ucker identity (cf.\ \cite{EGKT}).
\end{proof}

Now redefine $ m^{(2)}(k_1)$ according to
\[
 m^{(3)}(k_1)=m^{(2)}(k_1)\left\{\begin{array}{ll}
A_L(k_1), & k_1\in \Omega_l^U\cup\Omega_r^L,\\
 A_U(k_1)^{-1}, & k_1\in \Omega_l^L\cup\Omega_r^U,\\
B_L(k_1), & k_1\in \Omega_c^L,\\
 B_U(k_1)^{-1}, & k_1\in \Omega_c^U,\\
\id, & \mbox{else.}\end{array}\right.
\]
Then the vector function $ m^{(3)}(k_1)$ has no jump along $k_1\in\R$ and, by Lemma~\ref{plucker}, also not along $k_1\in[\I c, -\I c]$.
All remaining jumps on the contours $\mathcal C_l^L$, $\mathcal C_l^U$, $\mathcal C_c^L$, $\mathcal C_c^U$, $\mathcal C_r^L$, $\mathcal C_r^U$, and $\cup_{j=1}^N(\T_j^U\cup\T_j^L)$ are close to the identity matrix up to exponentially small errors except for small vicinities of the stationary phase points $k_{1}^-$ and $k_{1}^+$.
Thus, the model problem has the trivial solution $m^{\mathrm{mod}}(k_1)=\rI$. For large imaginary $k_1$ with $|k_1|>\kappa_{1,1}+1$ we have $m^{(2)}(k_1)= m^{(3)}(k_1)\sim m^{\mathrm{mod}}(k_1)$ and consequently
\[
m^{(1)}(k_1) = m^{(2)}(k_1)D^{-1}(k_1)= \begin{pmatrix} d^{(1)}(k_1)&  d^{(1)}(k_1)^{-1} \end{pmatrix}
\]
for sufficiently large $k_1$.
By \eqref{defdel}
\[
d^{(1)}(k_1)= 1 +\frac{1}{2\I k_1}\left(-\frac{1}{\pi}\int_{(-\infty,k_1^-) \cup(k_1^+,\infty)} \log(1-|R_1(s)|^2) ds\right) +O\left(\frac{1}{k_1^2}\right)
\]
and comparing this formula with formula \eqref{asm1} we conclude the expected leading asymptotics in the region $x<-c^2 t$ given by
\[
q(x,t)=c^2 (1 +O(t^{-1/2})).
\]
Moreover, the contribution from the small crosses at $k_1^\pm$ can be computed using the usual techniques \cite{dz}, \cite{GT}.

\begin{theorem}
In the domain $x<(-6c^2 -\epsilon)t$ the following asymptotics are valid:
\[
q(x,t)=c^2 + \sqrt{\frac{4\nu(k_1^+) k_1^+}{3t}}\sin(16t(k_1^+)^3-\nu(k_1^+)\log(192 t (k_1^+)^3)+\delta(k_1^+))+o(t^{-\alpha})
\]
for any $1/2<\alpha <1$.
Here $k_1^+=\sqrt{-\frac{c^2}{2}-\xi}$ and
\begin{align*}
\nu(k_1^+)= & -\frac{1}{2\pi} \log\left(1-|R_1(k_1^+)|^2\right),\\ \nn
\delta(k_1^+)= & \frac{\pi}{4}- \arg(R_1(k_1^+))+\arg(\Gamma(\I\nu(k_1^+)))\\
&  -\frac{1}{\pi}\int_{(-\infty,-k_1^+) \cup(k_1^+,\infty)}\log\left(\frac{1-|R_1(s)|^2}{1-|R_1(k_1^+)|^2}\right)\frac{1}{s-k_1^+}ds.
\end{align*}
\end{theorem}

The claim {\bf B} of Theorem \ref{maintheor} follows from this theorem by the change of variables $k_1\mapsto k$ and by use of \eqref{RTR1}.

\bigskip
\noindent{\bf Acknowledgments.} I.E. is indebted to the Department of Mathematics at the University of Vienna for its hospitality and support during the fall semester of 2015, where
this work was done. We are also indebted to the anonymous referee for numerous comments and suggestions which lead to a significant improvement of the presentation.

\end{document}